\documentclass{jfm}

\usepackage{graphicx}
\usepackage{natbib}
\usepackage{empheq}
\usepackage{relsize}

\ifCUPmtlplainloaded \else
  \checkfont{eurm10}
  \iffontfound
    \IfFileExists{upmath.sty}
      {\typeout{^^JFound AMS Euler Roman fonts on the system,
                   using the 'upmath' package.^^J}%
       \usepackage{upmath}}
      {\typeout{^^JFound AMS Euler Roman fonts on the system, but you
                   dont seem to have the}%
       \typeout{'upmath' package installed. JFM.cls can take advantage
                 of these fonts,^^Jif you use 'upmath' package.^^J}%
      }
  \else
  \fi
\fi

\ifCUPmtlplainloaded \else
  \checkfont{msam10}
  \iffontfound
    \IfFileExists{amssymb.sty}
      {\typeout{^^JFound AMS Symbol fonts on the system, using the
                'amssymb' package.^^J}%
       \usepackage{amssymb}%
         \let\leq=\leqslant
         \let\geq=\geqslant
      }{}
  \fi
\fi

\ifCUPmtlplainloaded \else
  \IfFileExists{amsbsy.sty}
    {\typeout{^^JFound the 'amsbsy' package on the system, using it.^^J}%
     \usepackage{amsbsy}}
    {\providecommand\boldsymbol[1]{\mbox{\boldmath $##1$}}}
\fi

\providecommand\bnabla{\boldsymbol{\bnabla}}
\providecommand\bcdot{\boldsymbol{\cdot}}
 
\newtheorem{thm}{Theorem}
\newtheorem{prop}{Proposition}
\newtheorem{lem}{Lemma}

\title[Geometric decomposition of the conformation tensor]
{Geometric decomposition of the conformation tensor in viscoelastic turbulence}

\author[Hameduddin, Meneveau, Zaki \& Gayme]%
{Ismail Hameduddin\aff{1,2} %
  \corresp{\email{ismailh@jhu.edu}}
  ,\ns Charles Meneveau\aff{1},\ns
Tamer  A. Zaki\aff{1} 
and Dennice  F. Gayme\aff{1,2} }

\affiliation{
\aff{1}Department of Mechanical Engineering, The Johns Hopkins University,
Baltimore, MD 21218, USA
\aff{2} Kavli Institute for Theoretical Physics, University of California, Santa Barbara, CA 93106, USA}

\usepackage{mathrsfs} 
\usepackage{amsmath,mathtools}
\usepackage[usenames,dvipsnames]{color} 

\newcommand{\sym}{{\textrm{sym}}}
\newcommand{\tr}{{\textrm{tr}}}
\newcommand{\Pos}{{\textbf{Pos}}}
\newcommand{\Sym}{{\textbf{Sym}}}
\newcommand{\GL}{{\textbf{GL}}}
\newcommand{\SO}{{\textbf{SO}}}
\newcommand{\dd}[2]{{ \frac{\text{D} {#1}}{\text{D} {#2}}   }} 
 
\newcommand{\I}[1]{\textrm{I}_{#1}}
\newcommand{\II}[1]{\textrm{II}_{#1}}
\newcommand{\III}[1]{\textrm{III}_{#1}}
\newcommand{\glaction}[2]{ \left[ #2 \right]_{ #1 } }
\newcommand{\bld}[1]{{\boldsymbol{ #1 }} } 
\DeclareMathAlphabet{\mathscrbf}{OMS}{mdugm}{b}{n} 
\newcommand\Wie{\mbox{\textit{Wi}}}  
\newcommand{\e}[1]{{\text{e}^{ #1 }} } 

\def\sldsh{\rule[0.2\baselineskip]{0.075in}{0.65pt}}
\def\sldot{\rule[0.2\baselineskip]{0.025in}{0.65pt}}
\def\bldot{\hspace{0.025in}}
\def\linesolids{\rule[0.2\baselineskip]{0.275in}{0.65pt}}
\def\linedashed{\sldsh\bldot\sldsh\bldot\sldsh} 
 
\def\linedotted{\sldot\bldot\sldot\bldot\sldot\bldot\sldot\bldot\sldot\bldot\sldot}

\def\sldshthick{\rule[0.2\baselineskip]{0.075in}{1.5pt}}

\def\linedashedthick{\sldshthick\bldot\sldshthick\bldot\sldshthick}

\def\linesolidsSym{\rule[0.2\baselineskip]{0.1375in}{0.65pt}\parbox{0.0825in}{$\vspace{-0.0075in}\mathlarger{\mathlarger{\mathlarger{\ast}}}$}\rule[0.2\baselineskip]{0.1375in}{0.65pt}}
\def\linedashedSym{\sldsh\bldot\sldsh\parbox{0.0825in}{{$\vspace{-0.0075in}\,\Box\,$}} \sldsh\bldot\sldsh} 
\def\linedshdotSym{\sldsh\bldot\sldot\bldot\parbox{0.0825in}{{$\vspace{-0.0075in}\mathlarger{\mathlarger{\mathlarger{\circ}}}$}}\bldot\sldsh\bldot\sldot} 

\begin{document}
  \maketitle 
 \begin{abstract}  
This work introduces a mathematical approach to analysing the polymer dynamics in turbulent viscoelastic flows that uses a new geometric decomposition of the conformation tensor, along with associated scalar measures of the polymer fluctuations.
The approach circumvents an inherent difficulty in traditional Reynolds decompositions of the conformation tensor: the fluctuating tensor fields are not positive-definite and so do not retain the physical meaning of the tensor.
The geometric decomposition of the conformation tensor yields both mean and fluctuating tensor fields that are positive-definite.
The fluctuating tensor in the present decomposition has a clear physical interpretation as a polymer deformation relative to the mean configuration. 
Scalar measures of this fluctuating conformation tensor are developed based on the non-Euclidean geometry of the set of positive-definite tensors. 
Drag-reduced viscoelastic turbulent channel flow is then used an example case study. The conformation tensor field, obtained using direct numerical simulations, is analysed using the proposed framework.
\end{abstract} 

\section{Introduction}

In the present study, we address the following questions that arise in the context of viscoelastic turbulent flows:
(a) given a turbulent flow whose dynamics are partially governed by state variables that are
positive-definite tensors representing material deformation, 
what is an appropriate method to decompose the flow into a mean, or
nominal, component and a deviation about that mean that preserves the
physical character of the state variables? and (b) are there corresponding
scalar measures of the turbulence associated with these positive-definite
state variables?
The conformation tensor is the relevant positive-definite state variable in viscoelastic turbulence.

Dilute polymer solutions, viscoelastic flows obtained by adding small amounts of polymers to an incompressible Newtonian solvent, are the focus of the present work.
The added polymers impart elasticity to the solvent which then causes the fluid to react not only to the deformation rate but also to the deformation history.
As a result, a complete physical description of viscoelastic turbulence requires characterization of both the velocity, $\bld{u}$, and the conformation tensor, $\mathsfbi{C}$, which together form the state variables.
The conformation tensor is a second-order positive-definite tensor that encapsulates the polymer deformation history and is obtained by averaging, over molecular realizations, the dyad formed by the polymer end-to-end vector \citep{Bird1987}.

The conformation tensor affects the velocity field through the polymer stress, $\mathsfbi{T} = \mathsfbi{T}(\mathsfbi{C})$, while gradients in the velocity field are responsible for polymer stretching.    
Characterizing the mean polymer stress, or stress deficit, was the focus of early work because it was found to be necessary for closing the mean momentum balance \citep{Willmarth1987}.  
A non-vanishing stress deficit suggests the possibility of maintaining a turbulent velocity profile in the absence of Reynolds stresses, in which case the polymer dynamics would sustain turbulence.
The experiments of \citet{Warholic1999} showed that a turbulent mean profile can, indeed, be maintained in the near absence of Reynolds stresses in channel flow. However, the polymer deformation itself is not readily accessible experimentally.
Therefore, much of the work in understanding the mechanisms that lead to behaviour such as that found by \citet{Warholic1999} has resorted to analytical treatments or direct numerical simulations (DNS), which we will briefly review below.

\subsection{Previous approaches to quantifying polymer deformation and its effects}

The main approach to analyse the polymer dynamics has been to utilize the statistics of the polymer forces and torques, or the normal stresses. The polymer force is the divergence of the polymer stress and the polymer torque is the curl of the polymer force.
For example,  \citet{DeAngelis2002} and \citet{Dubief2005} showed that cross-stream polymer force in turbulent channel flow counteracts spanwise variations in the velocity while enhancing streamwise advection, consistent with drag-reducing behaviour. 
The polymer torque acts in lockstep with the polymer force and counteracts streamwise vortices. It also inhibits generation of the heads of hairpin vortices \citep{Kim2007, Kim2013}. 
Recent theoretical work proposed a vorticity--polymer torque formulation of the linearised governing equations and used it to reveal a reverse Orr mechanism for turbulence production in viscoelastic parallel shear flows \citep{Page2014,Page2015}.
\citet{Min2003,Min2003a}, on the other hand, studied viscoelastic turbulent channel flow  but used the elastic energy, defined there as proportional to the sum of the normal polymer stresses, to posit a theory of drag reduction that relied on an active exchange of elastic and kinetic energies in the flow.

A more appropriate quantity to probe the polymer deformation itself, and one that is also a state variable, is the conformation tensor $\mathsfbi{C}$. The trace of $\mathsfbi{C}$, denoted here as $\tr\,\mathsfbi{C}$, is commonly used in the literature to analyse $\mathsfbi{C}$ since it is equal to the sum of its principal stretches and is therefore a measure of the polymer deformation. For example, \citet{Sureshkumar1997} considered first-order statistics of $\tr\,\mathsfbi{C}$ in their pioneering paper on the DNS of viscoelastic turbulent channel flows.
The quantity $\tr\,\mathsfbi{C}$ is frequently used because it is proportional to the elastic energy in purely Hookean constitutive models of the polymers \citep{Beris1994,Min2003a}.
However, it is often not a sufficiently complete descriptor of the polymer deformation; even if $\tr\,\mathsfbi{C}$ is held constant, the polymer may undergo a volumetric deformation.

\citet{Housiadas2003} evaluated $\tr\,\mathsfbi{C}$  for a wide range of flow parameters in a viscoelastic turbulent channel flow and found a surprising result for certain parameter ranges: the mean of $\tr\,\mathsfbi{C}$ can increase with increasing elasticity without a commensurate effect on the mean velocity profile.
A similar trend was also reported by \citet{Xi2010} in minimal flow unit simulations.
This trend is not inconsistent with the mean momentum balance as the mean of $\tr\,\mathsfbi{C}$ in turbulent channel flows can increase without effecting the mean velocity profile. This behaviour arises because the (mean) stress deficit is not a function of any of the normal components of $\mathsfbi{C}$.
The normal components of $\mathsfbi{C}$ do affect the mean momentum balance through the dynamical coupling between the different components of $\mathsfbi{C}$, but this relationship cannot be captured by $\tr\,\mathsfbi{C}$. 
The situation described above highlights the importance of simultaneously considering all of the components of $\mathsfbi{C}$ in order to arrive at a complete picture of the polymer deformation and its effect on the velocity field. 
Additionally, mean quantities such as those used by \citet{Housiadas2003} and others are in themselves insufficient descriptors of the fluctuating polymer deformation; higher-order statistical quantities associated with $\mathsfbi{C}$ are required to describe the fluctuations and their deviation away from the mean.

The fluctuating conformation tensor, $\mathsfbi{C}'$, and its various moments provide one method to obtain pertinent higher-order statistical descriptions of the full conformation tensor, $\mathsfbi{C}$  \citep[see also][]{Lee2017}. 
The tensor $\mathsfbi{C}'$ is obtained by subtracting the mean conformation tensor from $\mathsfbi{C}$, in analogy with the Reynolds decomposition of $\bld{u}$.
However, this fluctuating tensor is not guaranteed to be physically realizable; at least one realization of $\mathsfbi{C}'$ implies negative material deformation since it is guaranteed to lose positive-definiteness ($\text{tr}\,\mathsfbi{C}'$ must be $\leq 0$ for at least one sample pulled from a statistical ensemble).
Furthermore, it is not clear which scalar functions of $\mathsfbi{C}'$ provide mathematically consistent measures of the turbulence intensity associated with $\mathsfbi{C}$.

Although $\mathsfbi{C}'$ is not a physical conformation tensor, it can still be used for modelling Reynolds-averaged quantities in turbulence that do not necessarily require physically realizable fluctuating quantities.
Indeed, it has been used with varying degrees of success in recent work to develop turbulence models \citep{Masoudian2013,Resende2011,Iaccarino2010,Li2006}
and to quantify subgrid stress contributions in Large-Eddy Simulations (LES) \citep{Masoudian2016}. 
However, characterizing the polymer fluctuations using physically meaningful quantities is advantageous in that physical interpretations aid in modelling and provide a greater understanding of mechanism.

A physically motivated description of the velocity and conformation tensor field can be obtained using Karhunen-Lo\`{e}ve or proper orthogonal decomposition (POD), as recently shown by \citet{Wang2014}.
A POD is a global decomposition of a field quantity that yields an orthonormal basis that is optimally ordered in the sense of the best representation of the Euclidean norm \citep[see][for more details]{Lumley1970}.
For the velocity field this norm is the square-root of the kinetic energy but in a straightforward POD of the conformation tensor field the norm is not directly related to the elastic energy.
\citet{Wang2014} showed that a POD of the square-root of $\mathsfbi{C}$ instead ensures best representation in terms of the elastic energy. 
Their approach crucially assumed that $\text{tr}\,\mathsfbi{C}$ is proportional to the elastic energy. 
When this assumption is satisfied, their approach provides a valuable tool to extract the spatial structure of the dominant energetic components of viscoelastic turbulence. 
One can also use the individual modes to construct positive-definite tensors that represent modal polymer deformation, since the POD basis is orthogonal with respect to a Frobenius inner product integrated over the spatial domain. 
However, these tensors cannot be used to construct a local decomposition of the conformation tensor into mean and fluctuating components because the sum of squares is not equal to the square of the sum --- the cross contributions of the tensors only vanish when we take the trace and integrate over the spatial domain. 
Reynolds decomposing the square-root of the conformation tensor is a local approach in which the cross contribution similarly does not vanish instantaneously.

In the following subsection, we outline a local decomposition the conformation tensor into mean and fluctuating components, which overcomes the previously described difficulties with earlier approaches and which is one of the contributions of the present work. Another contribution is the development of scalar measures of the fluctuating component that depend on the Riemannian geometric structure of the set of positive-definite tensors. Although the results we invoke for the latter contribution are well-established and broadly applicable, they have not yet been widely used in fluid mechanics.

\subsection{The present approach: motivation and summary of the framework}

In the present work, we develop a formalism that allows us to evaluate the instantaneous deviation of the polymer deformation away from mean which respects the mathematical structure, and physical interpretation, of $\mathsfbi{C}$. 
Such a deviation yields an associated conformation tensor that can be used in analysis and modeling,  e.g.  the approach of \citet{Wang2014} can be adapted to the analysis of this new conformation tensor.
We also develop scalar measures of the turbulence intensity in the polymer deformation that reflect the distance, in the mathematically precise sense of a distance metric on a manifold, of the instantaneous deviation away from the mean.

In order to motivative our proposed formalism, we consider the following analogue encountered in the study of the stretching of material lines in turbulence \citep{Batchelor1952}. 
In this case, the normalized squared length of a material line, $\ell^2(t) > 0$,
serves as the scalar analogue of $\mathsfbi{C}$. 
Since material lines cannot vanish, $\ell^2 \neq 0$.
Let us assume that a statistically stationary state is possible and that $\langle\ell^2\rangle$ is then the expected value of the squared length of a material line.
A Reynolds decomposition of $\ell^2 = \langle \ell^2 \rangle + (\ell^2)'$ yields a fluctuation $(\ell^2)'(t)$ that is not always positive, which implies a negative normalized squared length.
One may, for the sake of argument, side-step the physical ambiguity implied by the negative squared length by considering only $|(\ell^2)'|$ but this does not solve the problem of asymmetry of $|(\ell^2)'|$ with respect the direction of the stretching; when $\ell^2/\langle\ell^2\rangle\in (1,\infty)$, the material line is expanded with respect to  $\langle\ell^2\rangle$ and when $\ell^2/\langle \ell^2 \rangle \in (0,1)$ it is compressed, which means that similarly probable states (expansion and contraction) would be described by fluctuations with very different magnitudes.
A meaningful way to study the fluctuations in $\ell^2$ is by instead considering $\log(\ell^2/\langle \ell^2\rangle)$.
Our goal is to generalize this latter type of construction to the conformation tensor, where one must take into account the tensorial nature of $\mathsfbi{C}$ which encodes directional information not included in a scalar such as $\ell^2$.

Following the scalar case described above, one approach to evaluate fluctuations in $\mathsfbi{C}$ is to use $\log\,\mathsfbi{C}$, in lieu of $\mathsfbi{C}$, where $\log$ here refers to the matrix logarithm. This approach is appealing because the logarithm of a positive-definite matrix is a symmetric matrix and the set of symmetric matrices form a vector space, which therefore allows for a Reynolds decomposition analogous to that of $\bld{u}$, i.e.
\begin{align}
\log\,\mathsfbi{C} = \langle \log\,\mathsfbi{C} \rangle + (\log\,\mathsfbi{C})',  \nonumber
\end{align}
where $\langle \log\,\mathsfbi{C} \rangle$ is the expected value of $\log\,\mathsfbi{C}$. 
To the authors' best knowledge, such a decomposition has not been previously used to characterize fluctuations in $\mathsfbi{C}$. However, $\log\,\mathsfbi{C}$ itself has been an object of some interest in the viscoelastic literature.
\citet{Fattal2004} introduced an approach for simulating viscoelastic flows that relied on evolving $\log\,\mathsfbi{C}$ instead of $\mathsfbi{C}$. \citet{Fattal2004} and \citet{Hulsen2005} then provided closed-form evolution equations for $\log\,\mathsfbi{C}$ which explicitly depended on both $\log\,\mathsfbi{C}$ as well as its spectral decomposition.
Recently, \citet{Knechtges2014} and \citet{Knechtges2015} eliminated the explicit dependence on the spectral decomposition but at the expense of either imposing restrictions on the spectral radius of $\mathsfbi{C}$ or introducing Dunford--Taylor-type integrals into the equations.

At least two additional difficulties arise in using $\log\,\mathsfbi{C}$. The first is that the expected value of $\mathsfbi{C}$ is not equal to $\e{\langle \log\,\mathsfbi{C} \rangle}$. This fact implies that evaluating the effect of the polymer stress on the mean momentum balance requires all statistical moments of $\log\,\mathsfbi{C}$, even when the polymer stress is a linear function of $\mathsfbi{C}$. The second difficulty is that, in general, $\e{ \langle \log\,\mathsfbi{C} \rangle + (\log\,\mathsfbi{C})'} \neq \e{ \langle \log\,\mathsfbi{C} \rangle } \bcdot \e{(\log\,\mathsfbi{C})'}$ which means that there is no way to associate $(\log\,\mathsfbi{C})'$ with a conformation tensor or a physical polymer deformation. It also means that there is no clear way to separate the effect of $\langle \log\,\mathsfbi{C} \rangle$ in the fluctuating momentum balance.  
 
In this paper, we derive a new conformation tensor, $\mathsfbi{G}$, from a physical decomposition of the polymer deformation.  Instead of the traditional additive decomposition 
\begin{align}
\mathsfbi{C} = \mathsfbi{\overline{C}} + \mathsfbi{C}', \label{CReyDecomp}
\end{align}
the proposed geometric decomposition of $\mathsfbi{C}$ is given by
\begin{align}
\mathsfbi{C} = \mathsfbi{\overline{F}} \bcdot \mathsfbi{G} \bcdot \mathsfbi{\overline{F}}^{\mathsf{T}},   \label{CGeoDecomp}
\end{align}
where $\mathsfbi{\overline{F}}$ is a deformation gradient tensor that can be calculated directly from a base-flow conformation tensor, $\mathsfbi{\overline{C}}$, such that $\mathsfbi{\overline{F}}\bcdot\mathsfbi{\overline{F}}^{\mathsf{T}} = \mathsfbi{\overline{C}}$ (this choice will be justified in \S \ref{sec:decomp} below).
This conformation tensor $\mathsfbi{G}$ is analogous to the scalar fluctuating quantity, $\ell^2/\langle \ell^2\rangle$. 

The tensor $\mathsfbi{G}$ represents turbulent deviations from the mean conformation tensor and can be analysed by resorting to the curved, Riemannian geometry of the manifold of positive-definite tensors.  Interestingly, the first two moment invariants of the tensorial equivalent of $\log(\ell^2/\langle \ell^2\rangle)$, i.e. $\log\,\mathsfbi{G}$, then appear as the relevant scalar measures for the fluctuations in $\mathsfbi{G}$. The first moment invariant is the logarithm of the ratio of the volume of $\mathsfbi{C}$ to the volume of $\mathsfbi{\overline{C}}$, where the volume of the conformation tensor refers to its determinant. The latter is proportional to the squared volume of the ellipsoid representing the coarse-grained polymer \citep{Truesdell2004}. The determinant also corresponds to the sphericity or conformational probability of the molecular structure \citep{Beris1994}. The second moment invariant is the metric distance of $\mathsfbi{G}$ away from $\mathsfbi{I}$ on the manifold of second-order positive-definite tensors. Finally,  we also propose a measure of the anisotropy of $\mathsfbi{C}$ relative to the mean, based on the work of \citet{Moakher2006}. This measure is equal to the metric distance of $\mathsfbi{G}$ to the closest isotropic tensor.

Finally, we use the proposed framework and direct numerical simulations to gain insight into the dynamics of viscoelastic (FENE-P) turbulent channel flow. Such flows are known to exhibit greatly reduced drag relative to an equivalent Newtonian flow, up to 60\% or more reduction in some cases \citep{Toms1948,White2008}. For such turbulent flows, separating the mean and fluctuating components of the conformation tensor in a physically consistent manner is an important step towards developing a quantitative understanding of the dynamics and isolating the relevant mechanisms at play.

The rest of the paper is organised as follows: 
the governing equations used in viscoelastic flows are reviewed in \S  \ref{sec:govEqn}.
Section \ref{sec:decomp} presents the geometric decomposition of the conformation tensor, the associated evolution equations, and the relation between the decomposition and the elastic energy.
A review of a geometry constructed specifically for the set of positive-definite tensors along with scalar measures of the polymer deformation based on this geometry and associated evolution equations are presented in \S \ref{sec:geodesicScalars}.  
In \S \ref{sec:DNSchannel}, we present an example case study of viscoelastic turbulent channel flow to illustrate the concepts developed in the paper.

\section{Governing equations}
 \label{sec:govEqn} 
The non-dimensional governing equations for the velocity, $\bld{u}$,  and conformation tensor, $\mathsfbi{C}$, in a viscoelastic flow are
\begin{align} 
\bnabla \bcdot \bld{u}  &= 0, \label{govEqn0} \\
\dd{\bld{u}}{t} +\bnabla p - \frac{\beta}{\Rey} \Delta \bld{u} - \frac{1-\beta}{\Rey} \bnabla \bcdot \mathsfbi{T} &= 0, \label{govEqn1}  \\
\dd{\mathsfbi{C}}{t} - 2\,\sym(\mathsfbi{C} \bcdot  \bnabla \bld{u})  
+ \mathsfbi{T} &= 0, \label{govEqn2}
\end{align}
where $\dd{(\cdot)}{t} = \p_t(\cdot) + u_k \p_k(\cdot)$ is the convective derivative, $\sym(\mathsfbi{A}) = \frac{1}{2}(\mathsfbi{A} + \mathsfbi{A}^{\mathsf{T}})$ is the symmetric part of a second-order tensor $\mathsfbi{A}$, $p$ is the pressure, $\Rey$ is the Reynolds number, $\beta \in [0,1]$ is the viscosity ratio and the polymer stress, $\mathsfbi{T}$,  is a function of the conformation tensor, $\mathsfbi{C}$.
The left-hand side of (\ref{govEqn2}) is equal to the upper-convected Maxwell derivative, or the Lie derivative with respect to $\bld{u}$, of $\mathsfbi{C}$. 
Although we restrict our focus to the upper-convected Maxwell derivative, it can be replaced in (\ref{govEqn2}) with any other co-rotational derivative, or objective rate.

The functional form  of $\mathsfbi{T}(\mathsfbi{C})$
 depends on the particular 
constitutive model and strain measure used. 
Although we will not use a particular model in the theoretical development that will follow, we note for completeness that in the absence of inherent directionality in the polymers,
$\mathsfbi{T}$ is an isotropic function of $\mathsfbi{C}$ defined locally at each $(\bld{x},t)$. Therefore, by the representation theorem \citep{Truesdell2004} we have
\begin{align}
\mathsfbi{T}(\mathsfbi{C}) = \frac{1}{\Wie} \left[\mu_0(\text{I}_{\mathsfbi{C}},
\text{II}_{\mathsfbi{C}},\text{III}_{\mathsfbi{C}}) \mathsfbi{I} + \mu_1(\text{I}_{\mathsfbi{C}},
\text{II}_{\mathsfbi{C}},\text{III}_{\mathsfbi{C}}) \mathsfbi{C} + \mu_2(\text{I}_{\mathsfbi{C}},
\text{II}_{\mathsfbi{C}},\text{III}_{\mathsfbi{C}}) \mathsfbi{C}^2\right],
\label{stressDefn}
\end{align}
where the three characteriztic tensor invariants of $\mathsfbi{C}$ are defined as
\begin{align}
 \I{\mathsfbi{C}} \equiv \tr\, \mathsfbi{C}, \quad
 \II{\mathsfbi{C}} 
\equiv
\frac{1}{2}\left[
(\tr\, \mathsfbi{C})^2 - \tr\,\mathsfbi{C}^2\right]
, \quad
 \III{\mathsfbi{C}} 
\equiv \text{det}\, \mathsfbi{C}. \label{Cinvariants}
\end{align} 
and the Weissenberg number $\Wie$ is the polymer relaxation time normalized by the convective time scale.

Table \ref{tab:modelCoeff} lists the coefficient functions, $\mu_i$ for $i=1,2,3$, for two
polymer models that are popular in the viscoelastic turbulence literature. The parameter $L_{\max}$ is the maximum polymer extensibility.


In the following (primarily in sections \ref{sec:govEqnsRANS}, \ref{sec:elasticEnergy}, and \ref{REScalars}), angled brackets, $\langle \cdot \rangle$, denote Reynolds spatio-temporal filtering  \citep[see appendix A in][]{Sagaut2006}, i.e. for a variable $\phi(\bld{x},t)$
\begin{align}
\langle \phi \rangle(\bld{x},t) = 
\int \phi(\bld{r},\tau) \mathscr{G}(\bld{x}-\bld{r},t-\tau)\text{d}^3\bld{r}\,\text{d}\tau \label{ReyFiltDefn}
\end{align}
where $\mathscr{G}$ is a filtering kernel that is normalized so that $\langle 1 \rangle = 1$, and is defined such that
\begin{align}
\langle \langle f_1 \rangle \rangle =  \langle f_1 \rangle, \qquad \langle \langle f_1 \rangle f_2 \rangle = \langle f_1 \rangle \langle f_2 \rangle \label{filteringProps}
\end{align}
for any two integrable functions $f_1 = f_1(\bld{x},t)$, $f_2 = f_2(\bld{x},t)$.
The mean of a quantity $\phi$ is then $\langle \phi \rangle$ and the $n$-th moment of $\phi$ is $\langle \phi^n \rangle$.  The properties, (\ref{filteringProps}), further imply that $\langle \mathscr{F} (\langle \phi \rangle) \rangle = \mathscr{F}(\langle \phi \rangle)$ for any analytic function, $\mathscr{F}$. While the example case study presented in \S \ref{sec:DNSchannel} uses traditional Reynolds time-averaging, we present definitions using the filtering formulation since the approach is also expected to be valid more generally.

We use an overlined symbol within the present text to denote the nominal or base-flow quantity associated with the symbol, which may be distinct from the averaged or filtered quantity.

\begin{table} 
\begin{center}
\begin{tabular}{c  c c c}
\hspace{0.5in}Model\hspace{0.5in} & $\mu_0$ & $\mu_1$ & $\mu_2$\\
\hline
Oldroyd-B & $-1$ & $1$ & $0$\\
FENE-P    & 
$\left[(3/L_{\max}^2) - 1\right]^{-1}$ & 
$ \left[1 - (\text{I}_{\mathsfbi{C}}/L_{\max}^2)\right]^{-1}$ & 
0  
\end{tabular}
\end{center}
\caption{Coefficients $\mu_i$ for two common models of polymers. Note that only $\mu_1$ in the FENE-P model depends on an invariant of $\mathsfbi{C}$.}
\label{tab:modelCoeff}
\end{table}

\section{Decomposition of the conformation tensor}  
 \label{sec:decomp}  
In the following, we will denote the general linear group of degree $n$, i.e. the set of $n \times n$ matrices with non-zero determinant, as $\GL_{n}$. We define the structure-preserving group action of $\GL_n$ on a set $\textbf{W}_n\subseteq\mathbb{R}^{n\times n}$ as
\begin{align}
\glaction{\mathsfbi{A}}{\mathsfbi{B}}  \equiv \mathsfbi{A}\bcdot\mathsfbi{B}\bcdot\mathsfbi{A}^{\mathsf{T}}.
\end{align}
where $\mathsfbi{A} \in \GL_n$ and $\mathsfbi{B} \in \textbf{W}_n$ and by definition, we require $\textbf{W}_n$ to be invariant under the action.

From the perspective of continuum mechanics, $\mathsfbi{C}>0$ is the left Cauchy-Green tensor associated with the deformation of the polymers \citep{Beris1994,Rajagopal2000,Cioranescu2016}, i.e. 
\begin{align} 
\mathsfbi{C} = \mathsfbi{F}\bcdot\mathsfbi{F}^{\mathsf{T}}= \glaction{\mathsfbi{F}}{\mathsfbi{I}},
\label{C_HtH}
\end{align}
where $\mathsfbi{F}$ is the deformation gradient 
with respect to an equilibrium configuration, also known as 
the distortion tensor.
If the spatial coordinates in the micro-structure are given 
by $\bld{a} = \bld{a}(\bld{a}_0,t)$ where $\bld{a}_0$ are the 
material coordinates, 
then $\mathsfbi{F} = \bnabla_{\bld{a}_0}\bld{a}= \partial \bld{a} /\partial \bld{a}_0$ so that 
a material line $\text{d}\bld{a}_0$ deforms to $\text{d}\bld{a} =   \mathsfbi{F} \bcdot \text{d}\bld{a}_0$ 
under the deformation represented by $\mathsfbi{C}$. 
When $\mathsfbi{F}$ is restricted to be symmetric, (\ref{C_HtH}) reduces to the factorization proposed by \citet{Balci2011} to improve numerical schemes for evolving the conformation tensor equations.

Let $\mathsfbi{\overline{C}}$ be a nominal conformation tensor such as the mean or laminar base-flow conformation tensor. The only requirement we impose on $\mathsfbi{\overline{C}}$ is that it must be defined according to a rule that ensures that $\mathsfbi{C}$ and $\mathsfbi{\overline{C}}$ cannot be arbitrarily rotated with respect to each other. In other words, if $\mathsfbi{C}$ transforms to $\glaction{\mathsfbi{R}}{\mathsfbi{C}}$ then $\mathsfbi{\overline{C}}$ must transform to $\glaction{\mathsfbi{R}}{\mathsfbi{\overline{C}}}$ for any $\mathsfbi{R}\in \SO_3$,  where $\SO_n$ denotes the $n \times n$ special orthogonal group (or rotation matrices). 
Define $\mathsfbi{\overline{F}} \in \GL_3$ with $\text{det}\, \mathsfbi{\overline{F}} >0$ as the tensor that satisfies 
\begin{align}
\mathsfbi{\overline{C}} =  \mathsfbi{\overline{F}} \bcdot  \mathsfbi{\overline{F}}^{\mathsf{T}}. \label{HbarDefn}
\end{align}
Such an $\mathsfbi{\overline{F}}$ is non-unique as it can be parameterized as
\begin{align}
\mathsfbi{\overline{F}} =   \mathsfbi{\overline{C}}^{\frac{1}{2}}\bcdot \mathsfbi{R} \label{HbarSol}
\end{align}
for any $\mathsfbi{R} \in \SO_3$ and where
$\mathsfbi{\overline{C}}^{\frac{1}{2}}$ is the unique matrix square-root of $\mathsfbi{\overline{C}}$.
Since the polar decomposition of $\mathsfbi{\overline{F}}$ and the square-root of $\mathsfbi{\overline{C}}$ (up to a $\pm$ sign change) are both unique, 
(\ref{HbarSol}) is a parametrisation of all possible $\mathsfbi{\overline{F}}$. The tensor $\mathsfbi{\overline{F}}$ serves as a deformation gradient associated with the mean configuration.

The $n$-th power of a positive-definite tensor $\mathsfbi{A}$ is a tensor with the same eigenvectors as $\mathsfbi{A}$ and associated eigenvalues equal to the corresponding eigenvalues of $\mathsfbi{A}$  raised to the $n$-th power.
In practice, since these $n$-th powers are isotropic functions of $\mathsfbi{A}$, one need not explicitly perform a spectral decomposition to calculate them.
For example, an application of the representation theorem can be used to express $\mathsfbi{ {A}}^{\frac{1}{2}}$ and $\mathsfbi{ {A}}^{-\frac{1}{2}}$ solely in terms of  $\mathsfbi{A}$ and its invariants \citep{Hoger1984,Ting1985}. 

\begin{figure}
\begin{center}
\includegraphics[draft=false,scale=0.7]{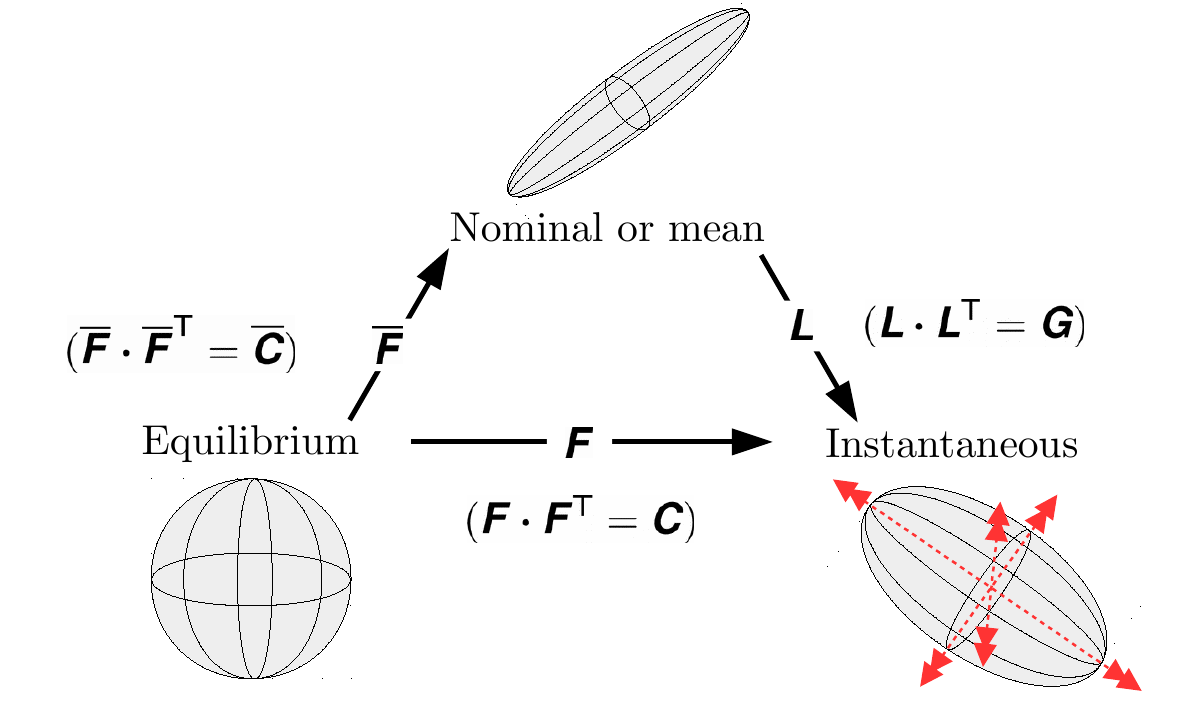}
\end{center}
\caption{Schematic of the decomposition given in (\ref{multiplicativeDecomp}) and (\ref{NecessaryDecomposition}).
The tensor $\mathsfbi{F}$ is a composition of $\mathsfbi{\overline{F}}$ and $\mathsfbi{L}$.} 
\label{fig:deformation_map}
\end{figure}

Given a specific $\mathsfbi{\overline{F}}$ one satisfying (\ref{HbarSol}), we then decompose the full distortion tensor $\mathsfbi{F}$ about $\mathsfbi{\overline{F}}$ by considering successive transformations on the material line $\text{d}\bld{a}_0$ , i.e.
\begin{align}
\text{d}\bld{a}= {\mathsfbi{F}}   \bcdot\text{d}\bld{a}_0 = \mathsfbi{\overline{F}} \bcdot {\mathsfbi{L}}  \bcdot\text{d}\bld{a}_0 \label{multiplicativeDecomp}
\end{align}
where $\mathsfbi{L}={\mathsfbi{\overline{F}}}^{\,-1}\bcdot  \mathsfbi{F}$  is the fluctuating distortion tensor.
This decomposition is illustrated in figure \ref{fig:deformation_map}.

Substituting  $\mathsfbi{F} = \mathsfbi{\overline{F}} \bcdot \mathsfbi{L}$ in (\ref{C_HtH}), we then arrive at a geometric decomposition of the conformation tensor
\begin{align}
\mathsfbi{C}  &=  
 \glaction{\mathsfbi{\overline{F}}}{\mathsfbi{G} }
 = \mathsfbi{\overline{F}}\bcdot \mathsfbi{G} \bcdot \mathsfbi{\overline{F}}^{\mathsf{T}}
\label{NecessaryDecomposition}
\end{align}
where $\mathsfbi{G} = \mathsfbi{L}\bcdot \mathsfbi{L}^{\mathsf{T}}$ is a left Cauchy-Green tensor that is analogous to $\mathsfbi{C}$. 
Comparing (\ref{NecessaryDecomposition}) and (\ref{CReyDecomp}), we can relate $\mathsfbi{C}'$ to $\mathsfbi{G}$ as follows
\begin{align}
\mathsfbi{C}'  &=  
 \glaction{\mathsfbi{\overline{F}}}{\mathsfbi{G} -\mathsfbi{I}}.
 \label{ReyGeomRel}
\end{align} 
From this point of view, the geometric decomposition provides a framework for interpreting the fluctuating tensor, $\mathsfbi{C}'$, obtained from the Reynolds decomposition.

Although a specific $\mathsfbi{G}$, and in particular only its set of principal axes, depends on $\mathsfbi{R}\in \SO_3$ chosen in (\ref{HbarSol}),  any function of only the invariants of $\mathsfbi{G}$ is independent of the choice of $\mathsfbi{R}$. This class of functions includes all objective scalar functions of $\mathsfbi{G}$; indeed, the scalar characterizations of the fluctuations that we develop later are also independent of $\mathsfbi{R}$. With respect to the full tensor, $\mathsfbi{G}$, we will later find that $\mathsfbi{R}=\mathsfbi{I}$ is a natural choice.

The decomposition of $\mathsfbi{F}$ into successive deformations, as in
(\ref{multiplicativeDecomp}), is reminiscent of the multiplicative decomposition in
large deformation theory that has found numerous applications over the last few
decades 
\citep{Casey2015,Sadik2017}. 
For example, in elasto-plasticity theory, 
the deformation gradient is decomposed into successive plastic and elastic
deformations with the objective of formulating constitutive laws for each of the
deformations somewhat independently. 
Similar constructions are used in thermo-elasticity and biomechanics
\citep{Lubarda2004}. A full review of that literature is beyond the scope of the
present work but it suffices to note that the present case is greatly simplified
because the constitutive laws are already specified and the focus is on the analysis
of the polymer deformation due to turbulence.

We next present the equations for mean and fluctuating quantities in the geometric decomposition when the nominal conformation tensor is obtained by averaging or Reynolds-filtering.

\subsection{Evolution equations in the Reynolds-filtered case}
\label{sec:govEqnsRANS}

In this section, we will consider the case when the nominal tensor is obtained using Reynolds-filtering. We choose to restrict our attention to Reynolds-filters but the development can be generalized to other filters,  e.g.  for applications in large-eddy simulations (LES).  We thus have
\begin{align}
\mathsfbi{\overline{C}} = \langle \mathsfbi{C}\rangle. \label{Cbardefn}
\end{align}
By the properties (\ref{filteringProps}), the associated $\mathsfbi{\overline{F}}$ satisfies $\langle \mathsfbi{\overline{F}} \rangle= \mathsfbi{\overline{F}}$.
Applying the averaging operation to (\ref{NecessaryDecomposition}) yields
\begin{align} 
\langle \glaction{\mathsfbi{R}}{\mathsfbi{G} } \rangle= \mathsfbi{I}
\end{align}
where $\mathsfbi{R}(\bld{x},t)$  is the rotation tensor field given in (\ref{HbarSol}). Henceforth, we will restrict the rotation tensor field so that
\begin{align}
\mathsfbi{R} = \langle \mathsfbi{R} \rangle.
\end{align}
By (\ref{filteringProps}), we then have $\langle \mathsfbi{G} \rangle = \mathsfbi{I}$. 

The Reynolds decomposition is applied to  $p$ and   $\bld{u}$ while $\mathsfbi{C}$ is decomposed using (\ref{NecessaryDecomposition}) with $\mathsfbi{\overline{C}}$ defined according to (\ref{Cbardefn}). We thus have
\begin{align}
p  &= \overline{p}+ p' , \quad
\bld{u}   = \bld{\overline{u}}+\bld{u}' ,\quad
\mathsfbi{C} =   \glaction{\mathsfbi{\overline{F}}}{\mathsfbi{G}} \label{decomp1}
\end{align} 
where  $\bld{\overline{u}} = \langle \bld{u} \rangle$ and  $\overline{p} = \langle p \rangle$ and the primes denote fluctuating quantities obtained via the Reynolds decomposition.  In general, $\overline{p} = \overline{p}(\bld{x},t)$, $\bld{\overline{u}} = \bld{\overline{u}}(\bld{x},t)$, $\mathsfbi{\overline{C}} = \mathsfbi{\overline{C}}(\bld{x},t)$. Note that $\mathsfbi{\overline{F}} \neq \langle \mathsfbi{F}\rangle$, in general.

Following the standard procedure, we can then decompose the momentum equation as follows
\begin{align}
\p_t\bld{\overline{u}} + \bld{\overline{u}}\bcdot \bnabla \bld{\overline{u}} &= - \bnabla \overline{p} +
\frac{\beta}{\Rey} \Delta \bld{\overline{u}} + \frac{1-\beta}{\Rey} \bnabla \bcdot \mathsfbi{\overline{T}} - \bnabla\bcdot \overline{ \bld{u}'\bld{u}'} \label{momentumEqn1}
\\
\p_t\bld{u}' 
+ \bld{\overline{u}}\bcdot \bnabla \bld{u}' + \bld{u}'\bcdot \bnabla \bld{\overline{u}}
&= - \bnabla p' + \frac{\beta}{\Rey} \Delta \bld{u}' + \frac{1-\beta}{\Rey} \bnabla \bcdot \mathsfbi{T}' 
- \bnabla \bcdot\left( \bld{u}'\bld{u}'\right)' \label{momentumEqn2}
\end{align}
where $\mathsfbi{\overline{T}} = \langle \mathsfbi{T} \rangle$, 
$\mathsfbi{T}'=\mathsfbi{T} - \mathsfbi{\overline{T}}$,
$\overline{ \bld{u}'\bld{u}'}=\langle \bld{u}'\bld{u}' \rangle$ and
$(\bld{u}'\bld{u}')' = \bld{u}'\bld{u}' - \overline{ \bld{u}'\bld{u}'}$.

The precise form of $\mathsfbi{\overline{T}}$, which appears in the mean momentum equation in (\ref{momentumEqn2}), depends on the constitutive model used. In the Oldroyd-B model, $\mathsfbi{\overline{T}}$ only depends on $\mathsfbi{\overline{F}}$:
\begin{align}
\mathsfbi{\overline{T}} = \frac{1}{\Wie}\left( \mathsfbi{\overline{F}}\bcdot \mathsfbi{\overline{F}}^{\mathsf{T}} - \mathsfbi{I} \right).
\label{meanStress_OLB}
\end{align}
In models that are nonlinear in $\mathsfbi{C}$, the fluctuating tensor $\mathsfbi{G}$ cannot be eliminated or factored out of $\mathsfbi{\overline{T}}$. For example, in the FENE-P model, $\mathsfbi{\overline{T}}$ can be expressed as a series in which the dominant term is equal to (\ref{meanStress_OLB}) while the remaining terms depend on higher-order moments of $\mathsfbi{G}$. In general, we have 
\begin{align}
\mathsfbi{\overline{T}} = \frac{1}{\Wie} \left[ \langle\mu_0\rangle  \mathsfbi{I} + \mathsfbi{\overline{F}}\bcdot \left\langle   \mu_1   \mathsfbi{G}   +  \mu_2    \mathsfbi{G} \bcdot \mathsfbi{\overline{F}}^{\mathsf{T}}\bcdot \mathsfbi{\overline{F}}\bcdot \mathsfbi{G} \right\rangle\bcdot \mathsfbi{\overline{F}}^{\mathsf{T}} \right].
\end{align}

Substituting (\ref{decomp1}) into (\ref{govEqn2}) and applying the filtering operation $\langle \cdot \rangle$ defined in (\ref{ReyFiltDefn}) yields the following equations for $\mathsfbi{\overline{C}}$  
\begin{multline}
\p_t \mathsfi{\overline{C}}_{ij} + \overline{u}_k\p_k \mathsfi{\overline{C}}_{ij} -
\left( \mathsfi{\overline{C}}_{ik} \p_{k}\overline{u}_j+ \mathsfi{\overline{C}}_{jk} \p_{k}\overline{u}_i \right) +  \mathsfi{\overline{T}}_{ij} 
=
-\p_k\Big[ 
\mathsfi{\overline{F}}_{ip} \mathsfi{\overline{F}}_{qj}^{\mathsf{T}}
\underbrace{\langle \mathsfi{G}_{pq} u_k' \rangle}_{\text{(a)}}\Big]\\
+
\mathsfi{\overline{F}}_{ip}\mathsfi{\overline{F}}_{qk}^{\mathsf{T}} 
\underbrace{\langle \mathsfi{G}_{pq} \p_{k}u_j'  \rangle}_{\text{(b)}} +
\mathsfi{\overline{F}}_{jp}\mathsfi{\overline{F}}_{qk}^{\mathsf{T}}
\underbrace{\langle \mathsfi{G}_{pq}  \p_{k}u_i' \rangle}_{\text{(c)}}.
\label{CbarEqn}
\end{multline}
Term (a) is the averaged turbulent transport and terms (b), (c) describe the mean stretching and rotation of the polymer arising due to the gradients in the fluctuating velocity field. 
The right-hand side of the equation (\ref{CbarEqn}) is the cumulative effect of the turbulent fluctuations on the mean balance. The mean balance (\ref{CbarEqn}) can also be written as
\begin{align} 
 \mathsfbi{\overline{M}}=
2\,\text{sym}\, \left\{\mathsfbi{\overline{E}}+ \langle \mathsfbi{G} \bcdot  \mathsfbi{E}'\rangle 
 - 
  \left[ \p_t \mathsfbi{\overline{F}}^{\mathsf{T}}  +
(\bld{\overline{u}} + \langle \mathsfbi{G}\bcdot\bld{u}' \rangle) \bcdot\bnabla \mathsfbi{\overline{F}}^{\mathsf{T}}  \right] \bcdot \mathsfbi{\overline{F}}^{-\mathsf{T}}
  \right\}  
-  \langle \bld{u}'\bcdot \bnabla \mathsfbi{G} \rangle  
\label{DHbarHbarInv}
\end{align}
where $\mathsfbi{\overline{E}} = \langle \mathsfbi{E}\rangle$, $\mathsfbi{E}'=\mathsfbi{E}-\mathsfbi{\overline{E}}$, $\mathsfbi{\overline{M}} = \langle \mathsfbi{M}\rangle$, and
\begin{align}
\mathsfbi{E}&\equiv \mathsfbi{\overline{F}}^{\mathsf{T}}\bcdot  \bnabla\bld{u} \bcdot \mathsfbi{\overline{F}}^{-\mathsf{T}} \\
\mathsfbi{M} &\equiv
\mathsfbi{\overline{F}}^{-1} \bcdot\mathsfbi{T} \bcdot \mathsfbi{\overline{F}}^{-\mathsf{T}}.
\label{GStress}
\end{align}
Here, the tensors $\mathsfbi{E}$ and $\mathsfbi{M}$ serve as modified velocity gradient and polymer stress tensors. Note that the invariants of $\bnabla\bld{u}$ and $\mathsfbi{E}$ coincide. The equation (\ref{DHbarHbarInv}) shows that the mean modified stress, $\mathsfbi{\overline{M}}$, is a function of the mean velocity gradient, $\mathsfbi{\overline{E}}$, the mean stretching due to turbulent velocity gradients $\langle \mathsfbi{G}\bcdot \mathsfbi{E}'\rangle$, and the turbulent advection of the fluctuating polymer deformation, $\langle \bld{u}'\bcdot \bnabla \mathsfbi{G} \rangle$. Additionally, new terms appear due to the time-rate of change of $\mathsfbi{\overline{F}}$, and modified advection of $\mathsfbi{\overline{F}}$.

%

We can find the evolution equation for $\mathsfbi{G}$ by substituting (\ref{decomp1}) into (\ref{govEqn2}) and simplifying, which yields 
\begin{align}
\dd{\mathsfbi{G}}{t}
=
2\, \sym(\mathsfbi{G} \bcdot \mathsfbi{K})
- \mathsfbi{M}
\label{GEvolEqn}
\end{align}
where $\mathsfbi{K}$ is given by
\begin{align}
\mathsfbi{K}  &\equiv 
\mathsfbi{E} 
- 
\left(\mathsfbi{\overline{F}}^{-1}\bcdot \frac{\text{D} \mathsfbi{\overline{F}}}{\text{D} t}\right)^{\mathsf{T}} 
 , 
\label{GVelGrad} 
\end{align}
and represents the modified velocity gradient augmented with an additional stretching that arises due to the decomposition.
The expression (\ref{GEvolEqn}) is more general than considered here; an equivalent expression can be derived when the nominal tensor, $\mathsfbi{\overline{C}}$, is not equal to the mean conformation tensor.
%
%

The higher-dimensional nature of $\mathsfbi{G}$ makes the quantification of the fluctuating turbulent polymer deformation a more difficult task.  We will examine the elastic potential energy as a method to evaluate this deformation in the next subsection and then introduce more general scalar characterizations of $\mathsfbi{G}$ in the next section.

 	\subsection{Elastic energy and its relation to $\mathsfbi{G}$}
 \label{sec:elasticEnergy}
 
The turbulent mean polymer configuration is not the thermodynamic equilibrium state, and thus $\mathsfbi{G}$ alone is not sufficient to fully determine thermodynamic quantities such as the elastic potential energy, $\varepsilon_{\psi}(\mathsfbi{C})$.  
For example, \citet{Beris1994} define $\varepsilon_{\psi}(\mathsfbi{C})$ for an Oldroyd-B model as
\begin{align}
\varepsilon_{\psi}(\mathsfbi{C}) =  \int_{\Omega} \psi\,  \text{tr}\,( \mathsfbi{\overline{C}}\bcdot\mathsfbi{G}) \, \text{d}^3 \bld{x}. \label{OLB_energy}
\end{align}
where we have rewritten the expression in terms of $\mathsfbi{G}$ and $\mathsfbi{\overline{C}}$  by setting $\mathsfbi{\overline{F}} = \mathsfbi{\overline{C}}^{\frac{1}{2}}$ and using the cyclic property of the trace to obtain $\text{tr}\,\mathsfbi{C}=\text{tr}\,(\mathsfbi{\overline{C}}\bcdot\mathsfbi{G})$. Here, $\Omega$ is the spatial domain, and the scalar function $\psi(\bld{x})$ is proportional to the polymer elastic constant times the elasticity density.

 The mean elastic potential, $\langle \varepsilon_{\psi}(\mathsfbi{C})\rangle$, for the Oldroyd-B model has the convenient property that it can be written solely in terms of the mean conformation tensor: $\langle \varepsilon_{\psi}(\mathsfbi{C})\rangle = \varepsilon_{\psi}(\mathsfbi{\overline{C}})$. However, the contribution of $\mathsfbi{G}$ in $\varepsilon_{\psi}(\mathsfbi{C})$ cannot be fully separated from that of $\mathsfbi{\overline{C}}$ because $\text{tr}\, \mathsfbi{A}\bcdot \mathsfbi{B} \neq \text{tr}\, \mathsfbi{A}\,\text{tr}\, \mathsfbi{B}$. Nonetheless, insight into the role of the different contributions can be obtained by using a trace inequality proved by \citet{Mori1988}, which yields
 	\begin{align}
 	\varepsilon_{\psi_3}(\mathsfbi{G}) 
 	\leq \varepsilon_{\psi}(\mathsfbi{C}) \leq  
 	\varepsilon_{\psi_1}(\mathsfbi{G}),
 	\qquad
 	\psi_i \equiv \psi\sigma_i(\mathsfbi{\overline{C}}), \label{OLB_energy_inequality}
 	\end{align}
 	where $\sigma_i(\mathsfbi{A})$ denotes the $i$-th largest eigenvalue of a tensor $\mathsfbi{A}$. In terms of the bounds in (\ref{OLB_energy_inequality}), the contribution of $\mathsfbi{\overline{C}}$ to $\varepsilon_{\psi}(\mathsfbi{C})$ is equivalent to a modification of the local elasticity density or the elastic constant.
 	
 	In other constitutive models, the contribution to the elastic potential energy from the mean polymer deformation is more difficult to separate. For example, the elastic potential energy for the FENE-P model \citep{Beris1994} is  
 	\begin{align}
 	\varepsilon_{\psi}(\mathsfbi{C}; L_{\max}) = -  \int_{\Omega} \psi L_{\max}^2 \log \left( 1 - \frac{ \text{tr}\,( \mathsfbi{\overline{C}}\bcdot\mathsfbi{G})}{L_{\max}^2}\right)\, \text{d}^3 \bld{x}, \label{FENEP_energy}
 	\end{align}
 	where $L_{\max}$, $\Omega$ and $\psi$ are as defined before. 
 	Here, the mean elastic potential energy, $\langle\varepsilon_{\psi}(\mathsfbi{C}; L_{\max})\rangle$, cannot be separated from $\mathsfbi{G}$ as in the Oldroyd-B model because, according to (\ref{FENEP_energy}), $\langle \varepsilon_{\psi}(\mathsfbi{C}) \rangle  \neq \varepsilon_{\psi}(\mathsfbi{\overline{C}})$. However, we can again bound $\varepsilon_{\psi}(\mathsfbi{C}; L_{\max})$ as
	\begin{align}  
\varepsilon_{\psi}(\mathsfbi{G}; L_{\max,3})
\leq
\varepsilon_{\psi}(\mathsfbi{C};L_{\max})
\leq
\varepsilon_{\psi}(\mathsfbi{G}; L_{\max,1}), \qquad 
L_{\max,i} \equiv L_{\max}/(\sigma_i(\mathsfbi{\overline{C}}))^{\frac{1}{2}}.   \label{FENEP_energy_inequality}
\end{align} 
 	In terms of the bounds in (\ref{FENEP_energy_inequality}), the contribution of $\mathsfbi{\overline{C}}$ in $\varepsilon_{\psi}(\mathsfbi{C};L_{\max})$ is equivalent to a modification of the local polymer maximum extensibility.
 	
 	Elastic energy may itself be insufficient to fully characterize the polymer deformation.
 	For example, in both Oldroyd-B and FENE-P models, the elastic energy is equal for all conformation tensors that are given by 
 	\begin{align}
 	\mathsfbi{C} = \glaction{\mathsfbi{Q}}{\text{diag}(\alpha_1 + \delta,\alpha_2-\delta,\alpha_3)}, \quad 0 \leq \delta < \alpha_2
 	\end{align}
 	where $\text{diag}(\phi_1,\phi_2,\phi_3)$ denotes a diagonal tensor with the $i$-th diagonal component given by $\phi_i$, and $\alpha_1 \geq \alpha_2 \geq\alpha_3 > 0$ and $\mathsfbi{Q}\in \SO_3$ are fixed. Even though the trace is fixed, the volume of the deformation ellipsoid changes with $\delta$, and is given by 
 	\begin{align} 
 	\text{det}\, \mathsfbi{C} =  \alpha_1\alpha_2\alpha_3+ (\alpha_1-\alpha_2)\alpha_3\delta +  \delta^2\alpha_3. \label{detExample}
 	\end{align}
 	In addition, since the governing equations are not Hamiltonian \citep{Beris1994}, the elastic potential energy only provides a partial characterization of the dynamics underlying the polymer deformation. Due to the above limitations of the elastic energy, and its dependence on the choice of the particular  constitutive model, we instead develop an approach to characterizing the polymer deformation using the inherent geometric structure underyling $\mathsfbi{G}$. This approach, introduced in the next section, is mathematically rigorous and can be applied to any positive-definite tensor.

\section{A Riemannian approach to the fluctuating polymer deformation}
\label{sec:geodesicScalars}

Any scalar characterization of $\mathsfbi{G}$ obeying the principle of objectivity can be a function only of its
invariants, $\I{\mathsfbi{G}}$, $\II{\mathsfbi{G}}$ and $\III{\mathsfbi{G}}$. The invariants can be
interpreted in terms of the fluctuating deformation ellipsoid, i.e. the ellipsoid
associated with $\mathsfbi{G}$.  The first invariant, $\I{\mathsfbi{G}}$, is proportional to the
average radius of the ellipsoid, the second invariant, $\II{\mathsfbi{G}}$, is proportional
to a lower bound for the surface area \citep{Klamkin1971}, and the third invariant
$\III{\mathsfbi{G}}$ is the volume of the deformation ellipsoid. Note that the eigenvalues (or principal stretches) of a conformation tensor are equal to the squared polymer stretches.

In practice, multiple difficulties arise in naively using the invariants of ${\mathsfbi{G}}$ to
characterize the conformation tensor. For example, consider the
isotropic case with $\mathsfbi{C} = a\mathsfbi{I}$ and $\mathsfbi{\overline{C}} = b\mathsfbi{I}$. We then
have $\mathsfbi{G} = (a/b)\mathsfbi{I}$ and the three invariants reduce to
\begin{align}
\I{\mathsfbi{G}} = 3a/b, \qquad
\II{\mathsfbi{G}} = 3(a/b)^2, \qquad
\III{\mathsfbi{G}} = (a/b)^3,
\end{align} 
which implies that the invariants are bounded between $0$ and $1$ for compressions with
respect to $\mathsfbi{\overline{C}}$ and between $1$ and $+\infty$ for expansions with
respect to $\mathsfbi{\overline{C}}$. This inherent asymmetry in the characterization is undesirable. The statistical moments of the invariants also vary over several orders of
magnitude, rendering these moments uninformative predictors of the level of turbulent
stretching in the polymers.

The problems discussed above arise because the set of $n\times n$ positive definite matrices, denoted $\Pos_n$, for $n > 0$, does not form a vector space and thus the Euclidean notions of translation and shortest distances between points are not
valid.
For example, let $\mathsfbi{A},\mathsfbi{B} \in \Pos_3$, and define $\mathsfbi{X}$ as \begin{align}
\mathsfbi{X} \equiv r \mathsfbi{A} + (1-r) \mathsfbi{B}, \qquad r\in \mathbb{R}.
\end{align}
One may wish to use the parameter $r$ to denote `distance of $\mathsfbi{X}$ to $\mathsfbi{A}$' along the `direction between $\mathsfbi{A}$ and $\mathsfbi{B}$'. However, $\mathsfbi{X}$ is then guaranteed to be positive-definite only if $r \in [0,1]$. While $\Pos_3$ is not a vector space, it has a Riemannian geometric
structure that can exploited to formulate alternative scalar measures of $\mathsfbi{G}$ that do not suffer from the problems mentioned above.  We introduce this geometry in \S \ref{sec:geodesic}, including definitions of shortest paths and distances between tensors.  Subsequently, in \S \ref{sec:scalars}, we introduce scalar measures based on the development in \S \ref{sec:geodesic} that can be used to quantify the turbulent fluctuations in the polymers. In \S \ref{REScalars}, we derive the Reynolds-filtered evolution equations for the scalar measures.

\subsection{Geodesic curves and distances between positive-definite tensors}
\label{sec:geodesic}

The set $\Pos_3$ is a Cartan-Hadamard manifold: it is a simply-connected, geodesically complete Riemannian 
manifold with seminegative curvature \citep{Lang2001}.
We summarize this characterization in the present section in order to develop a notion of distances between positive-definite tensors that will be used to formulate appropriate scalar measures of the fluctuating conformation tensor $\mathsfbi{G}$. Details on the Riemannian structure of $\Pos_3$ and theorems leading to the results used in
this section are presented in the appendix.

Consider two matrices $\mathsfbi{X},\mathsfbi{Y} \in \Pos_3$. In the set of all curves along the manifold $\Pos_3$
connecting $\mathsfbi{X}$ and $\mathsfbi{Y}$, there exists a unique curve that minimizes the distance between
$\mathsfbi{X}$ and $\mathsfbi{Y}$ with respect to the Riemannian metric on $\Pos_3$, i.e. there exists a $\mathsfbi{P}(r)$ with
$\mathsfbi{P}(0) = \mathsfbi{X}$ and $\mathsfbi{P}(1) =  \mathsfbi{Y}$ that uniquely minimizes the distance traversed along the
manifold between $\mathsfbi{X}$ and $\mathsfbi{Y}$. We call this curve the geodesic curve along the manifold and it is given by
\begin{align}
\mathsfbi{X}\#_r\mathsfbi{Y} =
\glaction{\mathsfbi{X}^{\frac{1}{2}}}
{ \left( \glaction{\mathsfbi{X}^{-\frac{1}{2}}}{\mathsfbi{Y}}\right)^r}, \qquad
0 \leq r \leq 1.
\end{align}
The geodesic distance associated with the geodesic curve between $\mathsfbi{X}$ and $\mathsfbi{Y}$ is the minimum separation between them along the manifold and is given by
\begin{align}
d(\mathsfbi{X},\mathsfbi{Y}) = \left[\sum_{i=1}^{3} (\log\,\sigma_i(\mathsfbi{X}^{-1}\bcdot\mathsfbi{Y}))^2\right]^{\frac{1}{2}} = \sqrt{\tr \log^2\left( \mathsfbi{X}^{-\frac{1}{2}}\bcdot\mathsfbi{Y}\bcdot \mathsfbi{X}^{-\frac{1}{2}}\right)}.
\label{distXY}
\end{align} 
The distance, $d(\mathsfbi{X},\mathsfbi{Y})$, is \emph{affine invariant}, i.e. $d(\mathsfbi{X},\mathsfbi{Y}) = d(\glaction{\mathsfbi{A}}{\mathsfbi{X}},\glaction{\mathsfbi{A}}{\mathsfbi{Y}}) $ for all $\mathsfbi{A} \in \GL_3$.

Geodesic curves and distances along a Riemannian manifold are analogous to straight lines and
distances in Euclidean space. In the case of $\Pos_3$ , the analogy can be taken quite far
because $\Pos_3$ is geodesically complete; a geodesic connecting any two points on the
manifold parameterized by $r$ can be arbitrarily extended letting $r \in \mathbb{R}$. For example, $\mathsfbi{X}\#_r\mathsfbi{Y}$ with $r \in [0,a]$, is a geodesic between $\mathsfbi{X}\#_0\mathsfbi{Y}$ and
$\mathsfbi{X}\#_a\mathsfbi{Y}$ for all $[0, a] \subseteq \mathbb{R}$. Furthermore, for each $r > 0$, we have
\begin{align}
d(\mathsfbi{X},\mathsfbi{X}\#_r\mathsfbi{Y}) = r\,d(\mathsfbi{X},\mathsfbi{Y}).
\end{align}  
We now illustrate the geodesic distance derived above using two specific examples. 
\begin{enumerate}
\item \emph{Isotropic tensors}:
Let $\mathsfbi{X} = a\mathsfbi{I}$ and $\mathsfbi{Y} = b\mathsfbi{I}$ be elements of the one-dimensional sub-manifold of $\Pos_3$ consisting of the isotropic tensors. The geodesic path joining $\mathsfbi{X}$ and $\mathsfbi{Y}$ is given by
\begin{align}
\mathsfbi{X}\#_r\mathsfbi{Y} = (a^{1-r}b^r)\mathsfbi{I} \label{isotropPath}
\end{align}
and the geodesic distance is given by 
$d(\mathsfbi{X},\mathsfbi{Y}) = \sqrt{3}\log\,(b/a)$.

Notice that $a^{1-r}b^r$, which appears in (\ref{isotropPath}), is a generalized geometric mean of $a$ and
$b$ with the classical definition realized at $r = 1/2$ . It can be shown that a similar
interpretation is admissible when $\mathsfbi{X}$ and $\mathsfbi{Y}$ are not isotropic \citep{Bhatia2015}. This fact
has formed the basis of efforts to formulate alternative definitions of statistical quantities
such as means and covariances so that they conform to the geometric structure of $\Pos_3$ \citep{Pennec2006,Fletcher2007}.

\item \emph{Tensors differing by a rotation}: Consider $\mathsfbi{X}$ and $\mathsfbi{Y} = \glaction{\mathsfbi{R}}{\mathsfbi{X}}$ for 
$\mathsfbi{R}\in\SO_3$. The geodesic joining $\mathsfbi{X}$ and $\mathsfbi{Y}$ is given by
\begin{align}
\mathsfbi{X}\#_r\mathsfbi{Y} = \glaction{\mathsfbi{X}^{\frac{1}{2}}}
{\left(\glaction{\mathsfbi{X}^{-\frac{1}{2}} \bcdot \mathsfbi{R}  }{\mathsfbi{X}}\right)^r}.
\end{align}
The distance between $\mathsfbi{X}$ and $\mathsfbi{Y}$ is then bounded as  
\begin{align}
0 \leq d(\mathsfbi{X},\mathsfbi{Y}) \leq
\sqrt{3} \min\,\left\{   
\max_i \left\{ \left| \log \left(\frac{\sigma_i(\mathsfbi{Y})}{\sigma_3(\mathsfbi{X})}\right) \right| \right\},
\max_i \left\{ \left| \log \left(\frac{\sigma_1(\mathsfbi{Y})}{\sigma_i(\mathsfbi{X})}\right) \right| \right\}
 \right\} \label{drotbnd}
\end{align}
where the lower bound is achieved for $\mathsfbi{R}=\mathsfbi{I}$. 
The upper bound in (\ref{drotbnd}) suggests that a differential rotation of a second-order tensor, $\mathsfbi{X}$, leads to an excursion along $\Pos_3$ with a path length that depends on the anisotropy of $\mathsfbi{X}$. For isotropic $\mathsfbi{X}$, the path length is zero and it otherwise increases with increasing anisotropy. 
 
\end{enumerate}

The geometry of $\Pos_3$ and the properties discussed above are next used to define scalar measures that characterize the turbulent fluctuations in $\mathsfbi{G}$.

\subsection{Scalar measures of the fluctuating conformation tensor}
\label{sec:scalars}
In this subsection we introduce scalar measures that can be used to quantify the fluctuating polymer deformation represented by $\mathsfbi{G}$.
In what follows, we will denote the matrix logarithm of $\mathsfbi{G}$ as $\bld{\mathcal{G}}$, i.e.
\begin{align}
\mathsfbi{G} 
=
\sum_{k=0}^{\infty}\frac{\bld{\mathcal{G}}^k}{k!}  
\equiv \text{e}^{\bld{\mathcal{G}}}.
\end{align}
The matrix logarithm is guaranteed to exist and is unique since $\mathsfbi{G}$ is positive-definite.  A key point to note is that the eigenvalues of $\bld{\mathcal{G}}$ are the logarithms of the eigenvalues of $\mathsfbi{G}$.
 
\subsubsection{Logarithmic volume ratio, $\zeta$}
Let $\Gamma_i = \sigma_i(\mathsfbi{G})$, for $i=1,2,3$, be the eigenvalues of $\mathsfbi{G}$. Then $\log\,(\text{det}\, \mathsfbi{G}) = \log ( \prod_{i=1}^3 \Gamma_i) = \sum_{i=1}^3 \log\Gamma_i$.
We thus define the logarithmic volume ratio of the fluctuation, $\zeta$, as
\begin{align}
\zeta &\equiv \tr\, \bld{\mathcal{G}} 
= \log( \text{det}\,\mathsfbi{G})
=  \log \left( \frac{\text{det}\,\mathsfbi{C}}{\text{det}\,\mathsfbi{\overline{C}} }\right).
 \label{zeta1defn}
\end{align} 
When $\zeta = 0$, the mean and the instantaneous conformation tensors have the same
volume; when $\zeta$ is negative (positive), the instantaneous conformation tensor has
a smaller (larger) volume than the volume of the mean. The logarithm ensures that there is no asymmetry between compressions and expansions with respect to the mean.

\subsubsection{Squared distance from the mean, $\kappa$}
When $\mathsfbi{C} = \mathsfbi{\overline{C}}$, we have $\mathsfbi{G}=\mathsfbi{I}$. When $\mathsfbi{C} \neq \mathsfbi{\overline{C}}$, we wish to consider the (appropriately defined) shortest distance between $\mathsfbi{I}$ and $\mathsfbi{G}$ as a measure of the magnitude of the fluctuation. The shortest path between $\mathsfbi{I}$ and $\mathsfbi{G}$ along the manifold $\Pos_3$ is given  by the geodesic along $\Pos_3$ connecting $\mathsfbi{I}$ and $\mathsfbi{G}$,
\begin{align}
 \mathsfbi{I} \#_r \mathsfbi{G}   =
  \mathsfbi{G}^r = \text{e}^{\bld{\mathcal{G}}r}.
\end{align}  
The squared geodesic distance associated with this path is then 
 \begin{align}
 \kappa &\equiv \tr\, \bld{\mathcal{G}}^2  =
 d^2(\mathsfbi{I},\mathsfbi{G}) =
 \sum_{i=1}^{3} (\log\,\Gamma_i)^2,
 \label{zeta2defn} 
 \end{align}
 where  (\ref{zeta2defn}) follows from (\ref{distXY}).
 Using (\ref{distXY}), one can verify that $d^2(\mathsfbi{I},\mathsfbi{G})=d^2(\mathsfbi{I},\mathsfbi{G}^{-1})$ and thus the squared distance measure treats both expansions and compressions with respect to the mean similarly. The affine-invariance property, furthermore, ensures that
 \begin{align}
 d^2(\mathsfbi{I},\mathsfbi{G})=d^2(\glaction{\mathsfbi{A}}{\mathsfbi{I}},\glaction{\mathsfbi{A}}{\mathsfbi{G}})
 \end{align}
 for all $\mathsfbi{A} \in \GL_3$. In particular, with $\mathsfbi{A} = \mathsfbi{\overline{F}}$, we obtain 
 \begin{align}
 d^2(\mathsfbi{I},\mathsfbi{G})
 =d^2(\mathsfbi{\overline{C}},\glaction{\mathsfbi{\overline{F}}}{\mathsfbi{G}})
 =d^2(\mathsfbi{\overline{C}}, \mathsfbi{C})=
 d^2(\mathsfbi{\overline{C}}^{-1}, \mathsfbi{C}^{-1})
 \end{align}
which exhibits the highly desirable property that the squared distance between $\mathsfbi{C}$ and $\mathsfbi{\overline{C}}$ is equal to the squared distance between $\mathsfbi{I}$ and $\mathsfbi{G}$. 
A further consequence of the affine-invariance property is that $d^2(\mathsfbi{I},\mathsfbi{G})$ is independent of the choice of the rotation $\mathsfbi{R} \in \SO_3$ in (\ref{HbarSol}).

The path between $\mathsfbi{\overline{C}}$ and $\mathsfbi{C}$ along $\Pos_3$ is given by 
\begin{align} 
  \mathsfbi{\overline{C}} \#_r \mathsfbi{C}  &= 
 \glaction{\mathsfbi{\overline{C}}^{\frac{1}{2}} }{
 \left( \glaction{\mathsfbi{\overline{C}}^{-\frac{1}{2}}}{\mathsfbi{C} } \right)^{r}},
\end{align}
which reduces to 
\begin{align}
  \mathsfbi{\overline{C}} \#_r \mathsfbi{C}  = \glaction{\mathsfbi{\overline{F}}}{\mathsfbi{G}^r}
\end{align}
when $\mathsfbi{R}=\mathsfbi{I}$ in (\ref{HbarSol}).
The choice $\mathsfbi{R}=\mathsfbi{I}$ is then natural in the sense that it allows the path along the manifold between
$\mathsfbi{\overline{C}}$ and $\mathsfbi{C}$, whose distance is a measure of the fluctuation, to be 
described using only $\mathsfbi{\overline{F}}$ and $\mathsfbi{G}$.

We next consider realizability in the $(\zeta,\kappa)$ plane.
Since $\bld{\mathcal{G}}$ is symmetric, its eigenvalues must be real, i.e. the eigenvalues must together belong in $\mathbb{R}^3$. 
In the $\mathbb{R}^3$ space of eigenvalues of $\bld{\mathcal{G}}$, surfaces of constant $\zeta$ are planes, and surfaces of constant $\kappa$ are spheres.
A particular choice of $\zeta$ and $\kappa$ is realizable only if the plane and sphere intersect.
The coordinates along the intersecting circle of the sphere satisfy
\begin{align}  
\cos\theta \sin \phi +
\sin \theta \sin \phi + 
\cos \phi = \frac{\zeta}{\sqrt{\kappa}}, \qquad
\theta \in [0,2\pi), \, \phi \in [0,\pi]
\label{admissible_zeta12_ratio}
\end{align}
where $\phi$ is the inclination angle and $\theta$ is the azimuthal angle, in a spherical coordinate representation of $\mathbb{R}^3$. The physically realizable region in the $(\kappa,\zeta)$ plane is thus given by
\begin{align}
-\sqrt{3 \kappa }  \leq
\zeta \leq 
\sqrt{3 \kappa } .
\label{admissible_zeta12_bounds}
\end{align}
When $\kappa = \frac{1}{3}\zeta^2$, the circle of intersection reduces to a point 
and $\bld{\mathcal{G}}$ consequently has only two independent tensor invariants, $\zeta$ and $\kappa$.
The angles that maximize the left-hand side of (\ref{admissible_zeta12_ratio}) are given by $\theta_{\max} = \pi/4$, $\phi_{\max} = \arctan \sqrt{2}$.
At $(\theta,\phi) = (\theta_{\max},\phi_{\max})$, it is readily verified that $\kappa = \frac{1}{3}\zeta^2$ and also that the eigenvalues of $\bld{\mathcal{G}}$ are all equal.
Thus $\bld{\mathcal{G}}$, and hence $\mathsfbi{G}$, is isotropic at the realizability bounds.

\subsubsection{Anisotropy index, $\xi$}
Following the approach taken by \citet{Moakher2006}, we define the anisotropy index, $\xi$, of $\mathsfbi{G}$ as the squared geodesic distance between $\mathsfbi{G}$ and the closest isotropic tensor, 
\begin{align}
\xi &\equiv \inf_{a} \, d^2(a\mathsfbi{I},\mathsfbi{G}) 
= 
\inf_{a} \,  
\tr\, (\bld{\mathcal{G}} - (\log\,a)\mathsfbi{I})^2. \label{chiDefn}
\end{align} 
By differentiation, we find that $a^3 = \prod_{i=1}^3 \sigma_i (\mathsfbi{G}) = \text{det}\, \mathsfbi{G}$ is a minimizing stationary point of (\ref{chiDefn}) and hence the closest isotropic tensor to $\mathsfbi{G}$ along $\Pos_3$ is $(\sqrt[3]{\text{det}\, \mathsfbi{G}})\mathsfbi{I}$. We then have
\begin{align}
\xi &= d^2((\sqrt[3]{\text{det}\, \mathsfbi{G}})\mathsfbi{I},\mathsfbi{G}) =   \kappa - \frac{1}{3} \zeta^2. \label{chiDefn1}
\end{align} 
Notice that $\chi=0$ if and only if $\zeta^2 = 3 \kappa$.
But since we already showed that $\bld{\mathcal{G}}$, and hence $\mathsfbi{G}$, are isotropic at the bound $\zeta^2 = 3 \kappa$, it follows that $\xi=0$ only for isotropic tensors.

\citet{Batchelor2005} first introduced the index $\sqrt{\xi}$ for characterizing positive-definite diffusion tensors measured in magnetic resonance imaging.
The index is analogous to the  `fractional anisotropy index' that is commonly used in turbulence and which provides the Euclidean distance to the closest isotropic tensor,
\begin{align}
\frac{\| \mathsfbi{G} - (\tr\, \mathsfbi{G}/3 )\mathsfbi{I} \|_F}{\| \mathsfbi{G} \|_F}, \label{fracIndex}
\end{align}
where $\| \mathsfbi{A} \|_F = \tr\,(\mathsfbi{A}^{\mathsf{T}}\bcdot \mathsfbi{A})$ indicates the Frobenius norm of matrix $\mathsfbi{A}$. 
A review of anisotropy measures is available in \citet{Moakher2006}. 

The three scalar measures presented above can be used together to obtain a better understanding of the fluctuations in the conformation tensor. The logarithmic volume ratio, $\zeta$, is positive (negative) for volumetric expansions (contractions) with respect to the mean. However, $\zeta=0$ does not necessarily imply no deformation, since  $\text{det}(\text{det}(\mathsfbi{G})\,\mathsfbi{A})=\text{det}(\mathsfbi{G})$ for all $\mathsfbi{A}$ with unit determinant. The squared geodesic distance to the identity, $\kappa$, helps  distinguish such cases since $\kappa=0$ only when $\mathsfbi{G}=\mathsfbi{I}$ ($\mathsfbi{C}=\mathsfbi{\overline{C}}$). Finally, the anisotropy index, $\xi$, provides a quantification of the deviation of the shape of the polymer from the shape of the mean conformation tensor because it is a measure of the distance from $\mathsfbi{G}$ to the closest isotropic tensor, or equivalently, the minimizing distance between $\mathsfbi{C}$ and $a\mathsfbi{\overline{C}}$ over all $a>0$.

We next derive evolution equations for the scalar measures presented above and for the particular case when $\mathsfbi{\overline{C}} = \langle \mathsfbi{C} \rangle$.

\subsection{Evolution equations for $\zeta$, $\kappa$ and $\xi$}
\label{REScalars}

 Since $\zeta$, $\kappa$ and $\xi$ are scalar characterizations of $\mathsfbi{G}$, one need only evolve $\mathsfbi{G}$ (or equivalently, $\mathsfbi{C}$) to obtain the field-valued $\zeta$, $\kappa$ and $\xi$. Nevertheless, it is of interest to  mathematically evaluate the evolution equations of these scalar measures separately in order to find the quantities that contribute to their dynamics. 

Using (\ref{GEvolEqn}) and the relationship
$\tr\,\bld{\mathcal{G}}^n = \sum_{i=1}^3 (\log \Gamma_i)^n$,   we can derive the following equations for the fluctuating scalar measures 
\begin{align}
\dd{\zeta}{t}
=   \tr \,  \bld{\mathcal{D}}, \qquad
\frac{1}{2}\dd{\kappa}{t}
 = 
 \tr\, (
 \bld{\mathcal{D}}\bcdot \bld{\mathcal{G}} ), \qquad
\frac{1}{2}\dd{\xi}{t}
= 
\tr\, ( \bld{\mathcal{D}} \bcdot \text{dev}\,\bld{\mathcal{G}}) \label{scalarGovEqns}
\end{align} 
where $\text{dev}\,\bld{\mathcal{G}}=\bld{\mathcal{G}} -  (\tr\, \bld{\mathcal{G}}/3) \mathsfbi{I}$ is the deviatoric part of $\bld{\mathcal{G}}$, and $\bld{\mathcal{D}}$ is defined as
\begin{align}
 \bld{\mathcal{D}} &\equiv 2\,\text{sym}\, \mathsfbi{K}  -  \mathsfbi{M}\bcdot\text{e}^{-\bld{\mathcal{G}}}.
\end{align} 
The derivation of (\ref{scalarGovEqns}), omitted here for brevity here, closely follows the procedure used by \citet{Vaithianathan2003} to obtain evolution equations for the continuous eigendecomposition of $\mathsfbi{C}$. 

The Cauchy-Schwarz inequality can be used to show that
\begin{align} 
\left|  \dd{ \sqrt{\kappa} }{t} \right|
\leq  \| \bld{\mathcal{D}} \|_{F}   , \qquad
  \left| \dd{ \sqrt{\xi} }{t} \right|
\leq  
\| \bld{\mathcal{D}} \|_{F}. \label{scalarUpperBnds}
\end{align} 
The bounds in (\ref{scalarUpperBnds}) illustrate the role of the stretching and relaxation balance, $\bld{\mathcal{D}}$, in bounding the growth of $\kappa$ and $\xi$.
In the Reynolds-filtered case, $\bld{\mathcal{D}}$ can be simplified using (\ref{DHbarHbarInv}) so that
\begin{multline}
\bld{\mathcal{D}}  =
2\,\text{sym}\, \left( \mathsfbi{E}'- \langle \mathsfbi{G} \bcdot  \mathsfbi{E}'\rangle 
 -      
 (\bld{u}' 
- \langle \mathsfbi{G}\bcdot\bld{u}'\rangle )\bcdot \bnabla   \mathsfbi{\overline{F}}^{\mathsf{T}} \bcdot \mathsfbi{\overline{F}}^{-\mathsf{T}} 
\right) 
  + 
\langle \bld{u}'\bcdot \bnabla \mathsfbi{G} \rangle   \\
-  \left( \mathsfbi{M}\bcdot\mathsfbi{G}^{-1}
-  \mathsfbi{\overline{M}}  \right) \label{Dsimplified}
\end{multline} 
which shows that the turbulence intensity of the fluctuating conformation tensor, as measured by $\kappa$, is not directly affected by the mean velocity gradient tensor $\bnabla\bld{\overline{u}}$. The contribution of $\bnabla\bld{\overline{u}}$ to $\kappa$ is captured indirectly through $\mathsfbi{\overline{F}}$, which is determined based on the mean balance.
 
According to (\ref{Dsimplified}), the tensor $\bld{\mathcal{D}}$ consists of a stretching component: $2\,\text{sym}\, \left( \mathsfbi{E}'- \langle \mathsfbi{G} \bcdot  \mathsfbi{E}'\rangle \right)$, a component that arises due to gradients in $\mathsfbi{\overline{F}}$ and represents modified advection of $\mathsfbi{\overline{F}}$:$ - 2\,\text{sym}\,\left[(\bld{u}' - \langle \mathsfbi{G}\bcdot\bld{u}'\rangle )\bcdot \bnabla   \mathsfbi{\overline{F}}^{\mathsf{T}} \bcdot \mathsfbi{\overline{F}}^{-\mathsf{T}}\right]$, a component that comprises mean advection of $\mathsfbi{G}$ by the fluctuating velocity field: $ \langle \bld{u}'\bcdot \bnabla \mathsfbi{G} \rangle$, and finally a component that resembles a fluctuating relaxation contribution: $-\left( \mathsfbi{M}\bcdot\mathsfbi{G}^{-1}
-  \mathsfbi{\overline{M}}  \right)$.

\subsubsection{Reynolds-filtering the evolution equations}
As a first-order statistical characterization of the fluctuating quantities, $\zeta$ and $\kappa$, we will consider their filtered or averaged values,
\begin{align}
\overline{\zeta} \equiv \langle \zeta \rangle, \qquad
\overline{\kappa} \equiv  \langle \kappa \rangle, \qquad
\overline{\xi}  \equiv \langle \xi \rangle.
\end{align}
Reynolds-filtering (\ref{scalarGovEqns}) and using the expression (\ref{Dsimplified}), we obtain the filtered evolution equations for $\zeta$, $\kappa$ and $\xi$.  
The filtered equation for $\zeta$ is given by
\begin{multline}
\left\langle\dd{\zeta}{t}\right\rangle    =
 -2\,\text{tr}\,\text{sym}\, \left( \langle \mathsfbi{G} \bcdot  \mathsfbi{E}'\rangle   -  
  \langle \mathsfbi{G}\bcdot\bld{u}'\rangle \bcdot \bnabla   \mathsfbi{\overline{F}}^{\mathsf{T}}\bcdot  \mathsfbi{\overline{F}}^{-\mathsf{T}}  \right)   \\
-\text{tr}\,\left(
-  \langle \bld{u}'\bcdot \bnabla \mathsfbi{G} \rangle  
+   \langle \mathsfbi{M}\bcdot\mathsfbi{G}^{-1}\rangle
-  \mathsfbi{\overline{M}}   \right)
\end{multline} 
where each term can be compared to those in (\ref{Dsimplified}) that were described in the previous subsection. Similarly, the filtered equation for $\kappa$ is given by 
\begin{multline}
\frac{1}{2}\left\langle\dd{\kappa}{t}\right\rangle   =
2\, \text{tr}\,\left\langle \text{sym}\, \left( \mathsfbi{E}'  
 -  \mathsfbi{\overline{F}}^{-1}\bcdot\bld{u}'\bcdot\bnabla
 \mathsfbi{\overline{F}} \right)\bcdot \bld{\mathcal{G}}  \right\rangle
-  \text{tr}\,\left\langle \mathsfbi{M}\bcdot\mathsfbi{G}^{-1} \bcdot \bld{\mathcal{G}} \right\rangle
\\
- \text{tr}\,\left\{ \left[
2\,\text{sym}\, \left( \langle \mathsfbi{G} \bcdot  \mathsfbi{E}'\rangle 
- 
  \langle \mathsfbi{G}\bcdot \bld{u}' \rangle \bcdot\bnabla\mathsfbi{\overline{F}}^{\mathsf{T}}\bcdot    \mathsfbi{\overline{F}}^{-\mathsf{T}} 
\right)  
 - \langle \bld{u}'\bcdot \bnabla \mathsfbi{G} \rangle   -\mathsfbi{\overline{M}} \right] 
\bcdot \langle \bld{\mathcal{G}}   \rangle  \right\}.
\end{multline} 
The equation for $\xi$, which we omit here for brevity, can be similarly derived.

\section{Case study: viscoelastic turbulent channel flow}
\label{sec:DNSchannel}

The general framework we have developed can be applied to a variety of flows. We focus on the classical problem of viscoelastic turbulent channel flow as a case study and use direct numerical simulations (DNS) to investigate the turbulent dynamics.  The algorithmic details of the simulation are identical to that of \citet{Lee2017} with the exception of the treatment of the conformation tensor which is documented in the appendix for the interested reader. The code employed was validated against linear growth rates of Tollmien-Schlichting waves  \citep[see][for a study of natural transition in viscoelastic flows]{Lee2017}  and also against the results of \citet{Agarwal2014}  for the evolution of a localized disturbance.

We define $x$, $y$ and $z$ as the streamwise, wall-normal and spanwise coordinates, respectively. The flow is homogeneous in $x$ and $z$ and all the coordinates are normalized with respect to the channel half-height, with the channel walls located at $y = \pm 1$. 
The parameters for the calculation are listed in Table \ref{tab:SimParams}. The Reynolds number, $\Rey$, is defined based on the channel half-height and bulk velocity while $\Rey_{\tau}$ is the friction Reynolds number defined based on the friction velocity, calculated using the slope of the mean velocity at the wall, and channel half-height.
The flow is driven by a pressure gradient which is adjusted in time to maintain a constant mass flow-rate. The symbol $\langle \cdot \rangle$ denotes averaging over $x$, $z$ and $t$. Therefore, all of the averaged quantities are functions of only $y$.

The computational grid is uniform in the $(x,z)$ directions and employs hyperbolic tangent stretching in the $y$-direction with a Planck taper \citep{McKechan2010} applied such that grid spacing very close to the wall is constant. The maximum change of the grid   spacing in the $y$-direction is less than $3$\% throughout in the domain. The resolution in friction units is listed in Table \ref{tab:SimParams}.  The initial turbulent state was generated from a separate simulation that followed the evolution of a Tollmien-Schlichting wave to the fully turbulent state \citep{Lee2017}. A snapshot from the fully turbulent state of \citet{Lee2017} was used as an initial condition and first run  for at least 150 convective time units before any statistics were collected. The evolution of the friction Reynolds number, $\Rey_{\tau}$, was used to check whether the simulation had reached a statistically stationary state.

\begin{table}
\begin{center}
\begin{tabular}{ccccccccc}
  &   &   &   &   & 
Domain size & Grid size & Time step 
& Spatial resolution \\
$\Rey$ & $\Rey_{\tau}$ & $\Wie$ & $L_{\max}$ & $\beta$ & 
$L_x\times L_y \times L_z$ & $N_x\times N_y \times N_z$ & $\Delta t$ 
& $\Delta_x^+ \times\Delta_y^+ \times\Delta_z^+$ \\
4667 & 180 & 6.67 & 100 & 0.9 & $4\pi \times 2 \times 4\pi$ &
$512 \times 400 \times 512$ &  $2.5 \times 10^{-3}$ & $4.42\times[0.13,1.90]\times 4.42$ 
\end{tabular}
\end{center}

\caption{Parameters of the simulation of viscoelastic turbulent channel flow.
The length scale is the channel half-height and velocity scale is the bulk flow speed. In the $p$-th directon, the size of the domain is $L_p$, the number of grid points is $N_p$, and the spatial resolution, in friction units, is $\Delta_p^+$.}
\label{tab:SimParams}
\end{table}

\subsection{Mean profiles and comparisons with the laminar profiles}
\begin{figure}
\centering
\includegraphics[draft=false,width=0.45\textwidth]{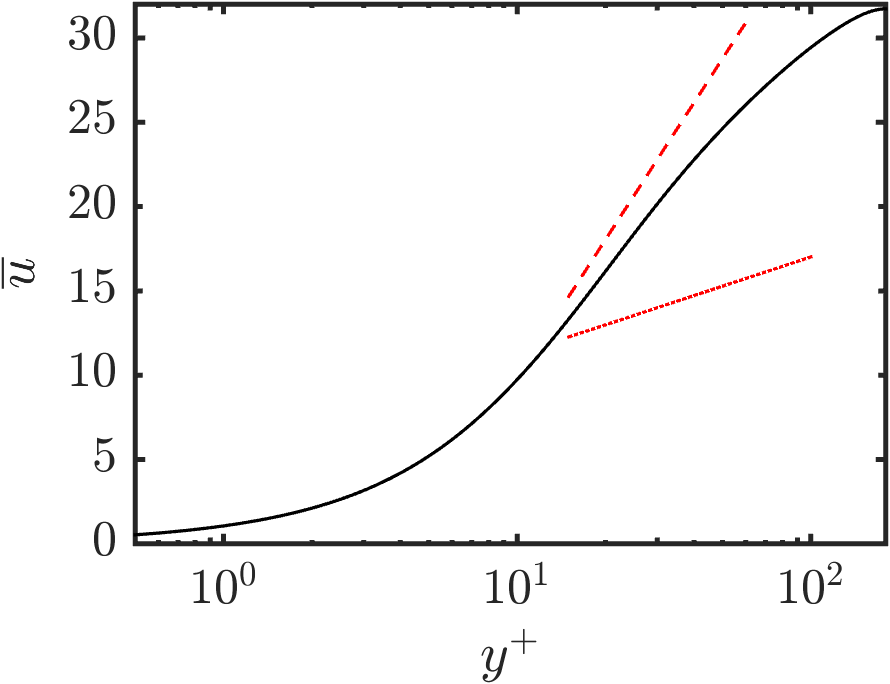}
\caption{Mean velocity profile from a FENE-P drag-reduced channel flow simulation.  The solid
line (\linesolids) is the mean streamwise velocity, the red dotted (lower) line ({\color{red}\linedotted}) 
is the von K\'{a}rm\'{a}n log-law, $\overline{u}^+_{\text{von K\'{a}rm\'{a}n}} = 2.5y^+ + 5.5$,
and the red dashed (upper) line ({\color{red}\linedashed}) is Virk's asymptote, $\overline{u}^+_{\text{Virk}} = 11.7y^+-17.0$. }
\label{fig:u-mean-profile}
\end{figure}

\begin{figure}
\hspace{0.6in} (a) \hspace{1.8in} (b)

\centering
\includegraphics[draft=false,width=0.75\textwidth]{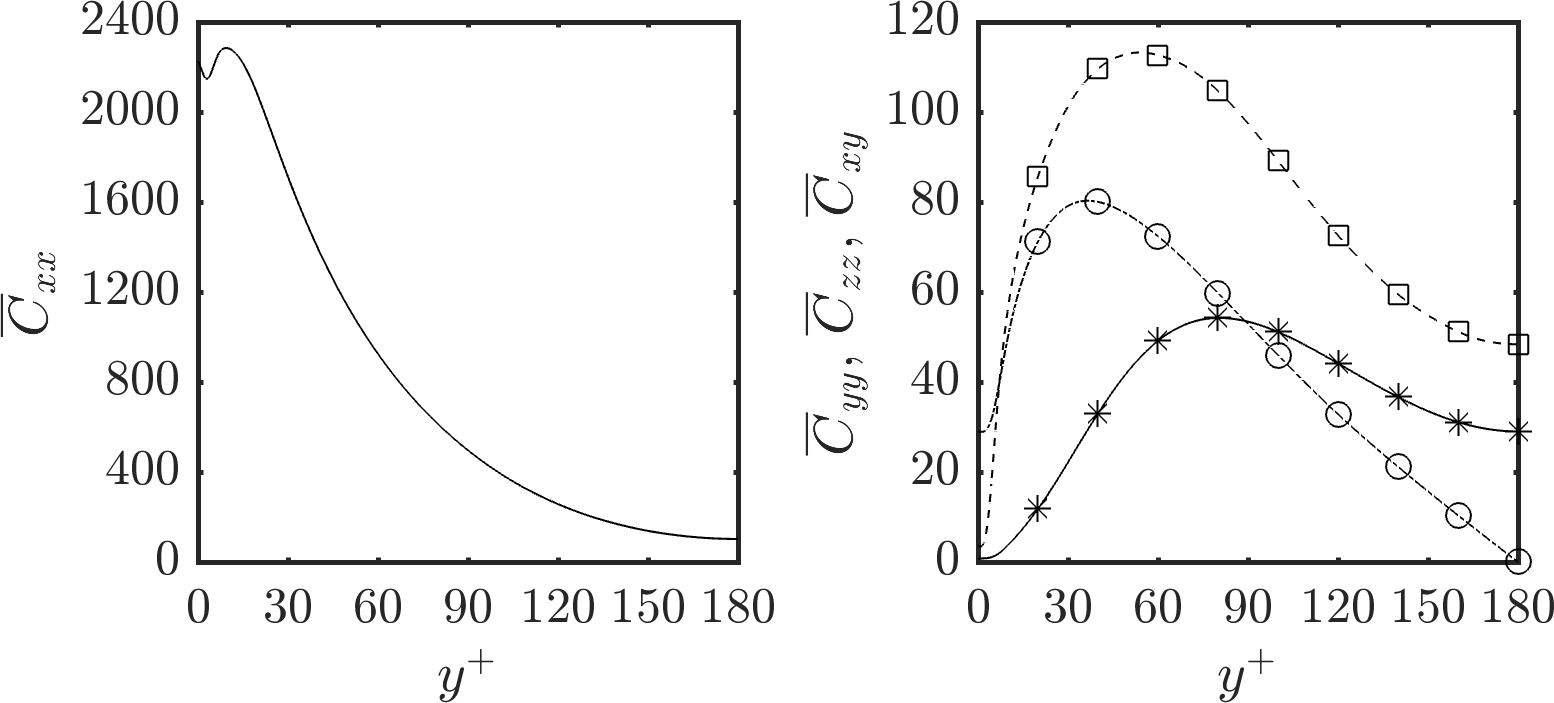}
\caption{Mean conformation tensor profiles from a FENE-P drag-reduced channel flow simulation. (a) $\mathsfi{\overline{C}}_{xx}$ (b) the solid line with star symbols (\linesolidsSym) is $\mathsfi{\overline{C}}_{yy}$, the dashed line with square symbols  (\linedashedSym) is $\mathsfi{\overline{C}}_{zz}$, and the dashed-dot line with circle symbols
 (\linedshdotSym) is $\mathsfi{\overline{C}}_{xy}$. The remaining components of the mean conformation tensor are $0$. Note that the symbols in (b) are identifiers and are thus only a small subset of all the data points used in the line plots. } 
\label{fig:C-mean-profiles}
\end{figure} 

The statistics presented in this section were obtained by averaging in space and over $750$ time units, and by exploiting the symmetry of the flow about the centreline. Halving the number of samples maintained the trends and caused only minor deviations in the statistics, with no impact on the conclusions.

The mean streamwise velocity profile is shown in figure \ref{fig:u-mean-profile} as a function of $y^+=\Rey_{\tau}(y+1)$, where $\Rey_{\tau}$ is always taken to be the turbulent frictional Reynolds number given in Table \ref{tab:SimParams}. Also shown are the 
von K\'{a}rm\'{a}n log-law and Virk's maximum drag reduction asymptote. The mean velocity lies 
in between these two lines, indicating a drag-reduced state. Using Dean's correlation
for the skin-friction \citep{Dean1978}, we obtain a friction Reynolds number of approximately $284$
for a Newtonian flow with $\Rey = 4667$ and thus the drag reduction percentage is  
\begin{align}
\text{DR}\% \equiv \left[1 - \left(\frac{\Rey_{\tau}}{\Rey_{\tau}|_{\text{Newtonian}}}\right)^2\right]\times 100 = 59.8\%.
\end{align} 

The non-zero components of $\langle\mathsfbi{C}\rangle$  calculated for the same parameter values, are shown in figure
\ref{fig:C-mean-profiles}. All the components of $\langle\mathsfbi{C}\rangle$ are even functions of $y$
except $\langle \mathsfi{C}_{xy} \rangle$, which is an odd function of $y$. The streamwise stretch $\langle \mathsfi{C}_{xx} \rangle$ is four times larger than the laminar case (not shown) near the wall. It is also an order of
magnitude larger than $\langle \mathsfi{C}_{yy} \rangle$ and $\langle \mathsfi{C}_{zz} \rangle$. 
The remaining normal components of the conformation tensor are also larger in the turbulent case than the laminar:  figure \ref{fig:C-mean-profiles} shows that
$\max_y \langle \mathsfi{C}_{yy} \rangle \approx 45$ and $\max_y \langle \mathsfi{C}_{zz} \rangle \approx 120$, while
$\mathsfi{C}_{yy}=\mathsfi{C}_{zz}\approx 1$ throughout the channel when the flow is laminar. The peak values of each
of the components $\langle \mathsfi{C}_{xx} \rangle$, $\langle \mathsfi{C}_{yy} \rangle$ and $\langle \mathsfi{C}_{zz} \rangle$
occur at different locations in the channel. Figure \ref{fig:C-mean-profiles} also shows that
$\langle \mathsfi{C}_{zz} \rangle \geq \langle \mathsfi{C}_{yy} \rangle$ throughout the channel. 
The trends above are consistent with those reported in the literature \citep{Dallas2010}.

\begin{figure}
\hspace{0.6in} (a) \hspace{1.8in} (b)

\centering
\includegraphics[draft=false,width=0.75\textwidth]{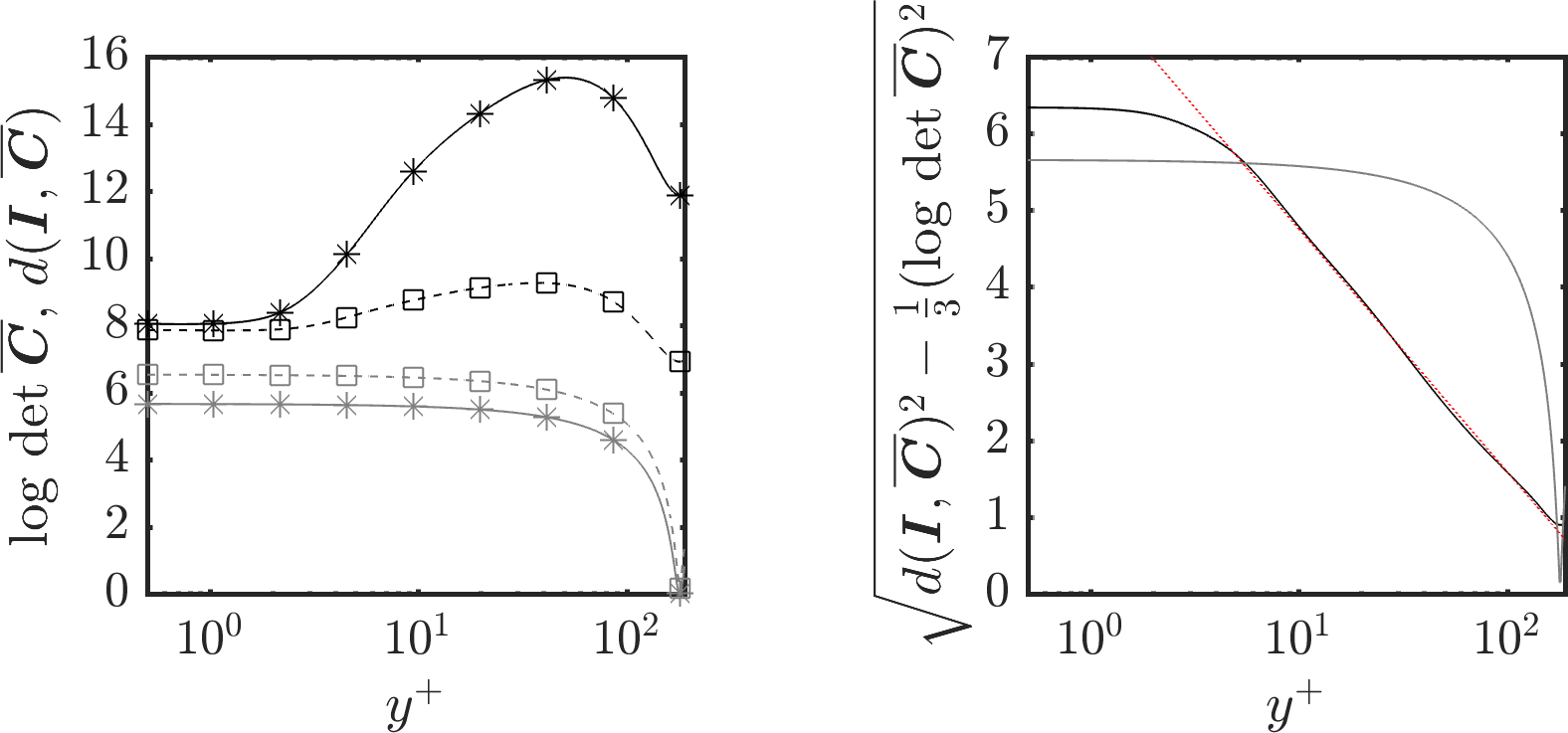}
\caption{Scalar measures applied to nominal conformation tensors, in equivalent dimensions, plotted as functions of $y^+$. (a) The solid line with star symbols (\linesolidsSym) is $\log\,\det\,\mathsfbi{\overline{C}}$, the logarithmic volume;
the dashed line with square symbols (\linedashedSym) is $d(\mathsfbi{I},\mathsfbi{\overline{C}})$, the geodesic distance between $\mathsfbi{\overline{C}}$ and $\mathsfbi{I}$ on $\Pos_3$, (b) anisotropy index, $\sqrt{(d(\mathsfbi{I},\mathsfbi{\overline{C}}))^2 - \frac{1}{3} (\log\,\det\,\mathsfbi{\overline{C}})^2}$, which is the geodesic distance from the closest isotropic tensor.
For both (a) and (b), black lines are for $\mathsfbi{\overline{C}} = \langle \mathsfbi{C} \rangle$ and grey lines are when 
$\mathsfbi{\overline{C}}$ is equal to the FENE-P laminar conformation tensor. The red dotted line ({\color{red}\linedotted}) in (b) is 
$-1.375\log\,y^+ + 7.925$. Note that the symbols in (a) are identifiers and are thus only a small subset of all the data points used in the line plots.}
\label{fig:lgCbar-profiles}
\end{figure}

Figure \ref{fig:lgCbar-profiles}(a) shows the logarithmic volume of the mean and laminar conformation tensors along with the distance from the origin on the manifold of positive-definite tensors. 
The figure shows that both the logarithmic volume,
$\log\,\det\,\mathsfbi{\overline{C}}$, and the distance from the origin, 
$d(\mathsfbi{I}, \mathsfbi{\overline{C}})$ are larger in the turbulent case compared to the laminar.
Furthermore, these two quantities are monotonically decreasing in the laminar case but have peaks around $y^+ \approx 60$ in the turbulent case. Both quantities in the two cases asymptote to a constant at locations very close to the wall, $y^+ \leq 2$. 
The weak growth in $d(\mathsfbi{I}, \mathsfbi{\overline{C}})$ in the turbulent case despite a rapid increase in $\log\,\det\,\mathsfbi{\overline{C}}$ is due to the increase in isotropy (sphericity) as we move away from the wall, since for a given volume the tensor closest to $\mathsfbi{I}$ is an isotropic tensor. In figure \ref{fig:C-mean-profiles} we see that the mean normal stretches in the $y$ and $z$
directions in the turbulent case peak between $y^+ \approx 40$ and $y^+ \approx 70$ and
are at least an order of magnitude larger than the stretches in the laminar case, where $\mathsfi{\overline{C}}_{yy} = \mathsfi{\overline{C}}_{zz} \approx 1$. The term $\mathsfi{\overline{C}}_{xx}$
is decreasing towards the channel centre in both the turbulent and laminar case but is accompanied by an increase in $\mathsfi{\overline{C}}_{yy}$,
$\mathsfi{\overline{C}}_{zz}$ in the turbulent flow which leads to increased isotropy for locations sufficiently removed from the wall. 

Figure \ref{fig:lgCbar-profiles}(b) shows the anisotropy, the geodesic distance to the closest isotropic tensor, of the mean and laminar conformation tensors. 
The anisotropy index is approximately constant in the vicinity of the wall for both the
laminar and turbulent cases, and decays away from the wall. In the turbulent flow, the decay starts very close to the wall --- approximately three friction units away from the wall, and then shows a remarkable logarithmic decay that proceeds all the way to very close to the centreline where it sharply turns and forms a stationary point.  The increased isotropy in the turbulent case may be explained by the
fact that, although more stretching occurs in this case, the stretching in the cross-stream
directions is much larger than in the laminar case. Overall, this leads to a more isotropic mean conformation tensor.

\subsection{Invariants of the fluctuating conformation tensor}

\begin{figure} 
\hspace{0.8in} (a) \hspace{1.6in} (b)

\begin{center}
\includegraphics[draft=false,width=0.8\textwidth]{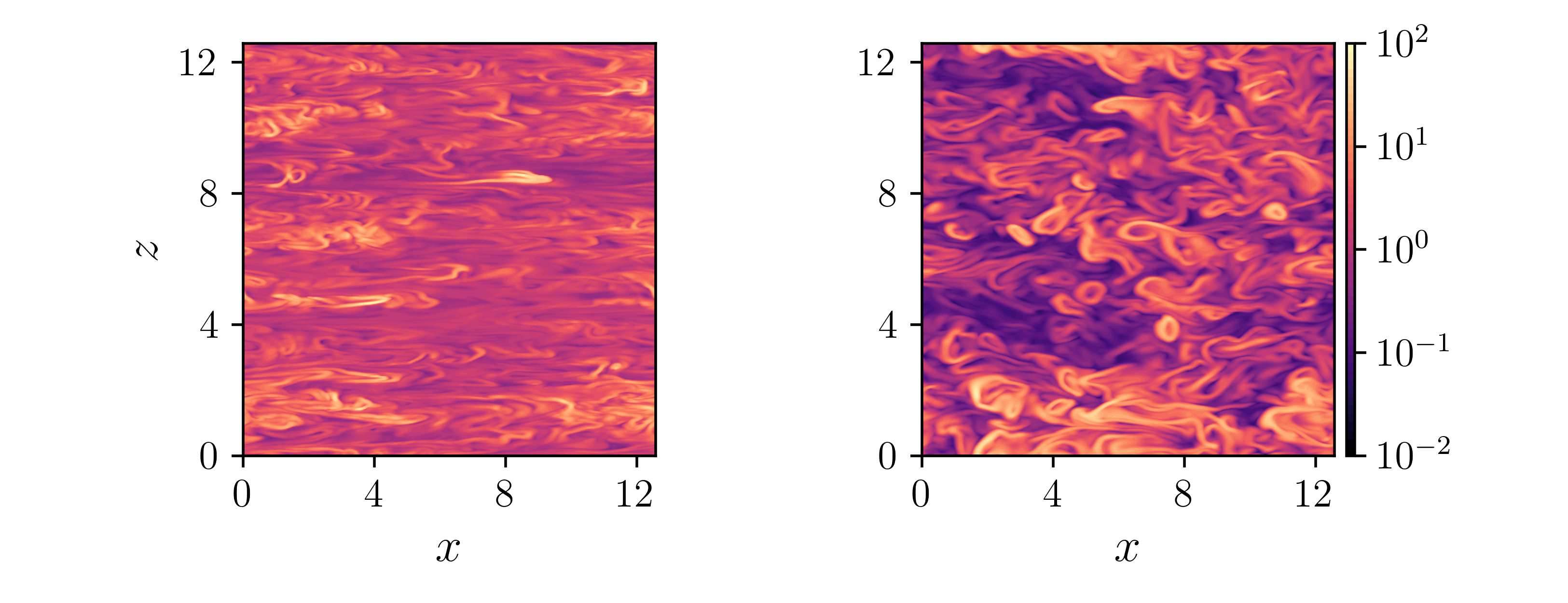} 
\end{center} 

\caption{Wall-parallel $(x,z)$ planes of isocontours of instantaneous $\I{\mathsfbi{G}} = \text{tr}\,\mathsfbi{G}$.
(a) at $y^+=15$ and (b) $y^+=180$ (centreline).}
\label{fig:rho-snapshots}
\end{figure}

In order to motivate the scalar measures proposed in the present work, we consider the invariants of $\mathsfbi{G}$ as alternatives in this subsection.
Figure \ref{fig:rho-snapshots} shows isocontours of instantaneous $\I{\mathsfbi{G}}$ at a given time at two wall-parallel planes, $y^+ = 15$ and $y^+ = 180$ (centreline). 
The isocontours of instantaneous $\II{\mathsfbi{G}}$ and $\III{\mathsfbi{G}}$ are qualitatively similar to those of $\I{\mathsfbi{G}}$ and are thus not shown here. The instantaneous $\I{\mathsfbi{G}}$ can vary over several orders of magnitude. As a result, obtaining reliable statistics for the
invariants is challenging. We found that the peak
root-mean-square (RMS) of the invariants (not shown) are at least an order of magnitude larger than the corresponding mean values.
This large spread in the instantaneous invariants of $\mathsfbi{G}$ suggests that $\log\,\mathsfbi{G}$ is a more
appropriate quantity to consider, and reinforces the need for the geometrically consistent scalar measures introduced in
\S \ref{sec:geodesicScalars}.

\subsection{The scalar measures: $\zeta$, $\kappa$ and $\xi$}

\begin{figure}

\hspace{0.8in} (a) \hspace{1.6in} (b)

\begin{center}
\hspace{0.1in}\includegraphics[draft=false,width=0.8\textwidth]{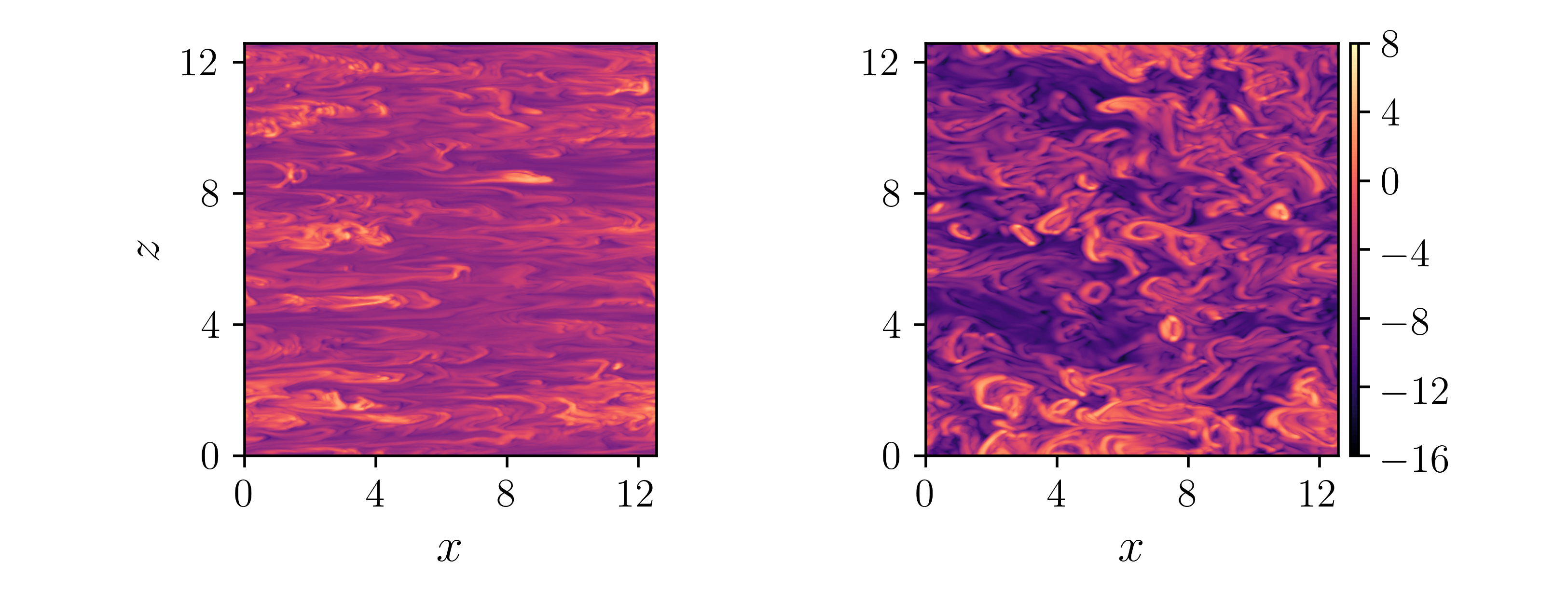}
\end{center}

\hspace{0.8in} (c) \hspace{1.6in} (d) 

\begin{center}
\includegraphics[draft=false,width=0.8\textwidth]{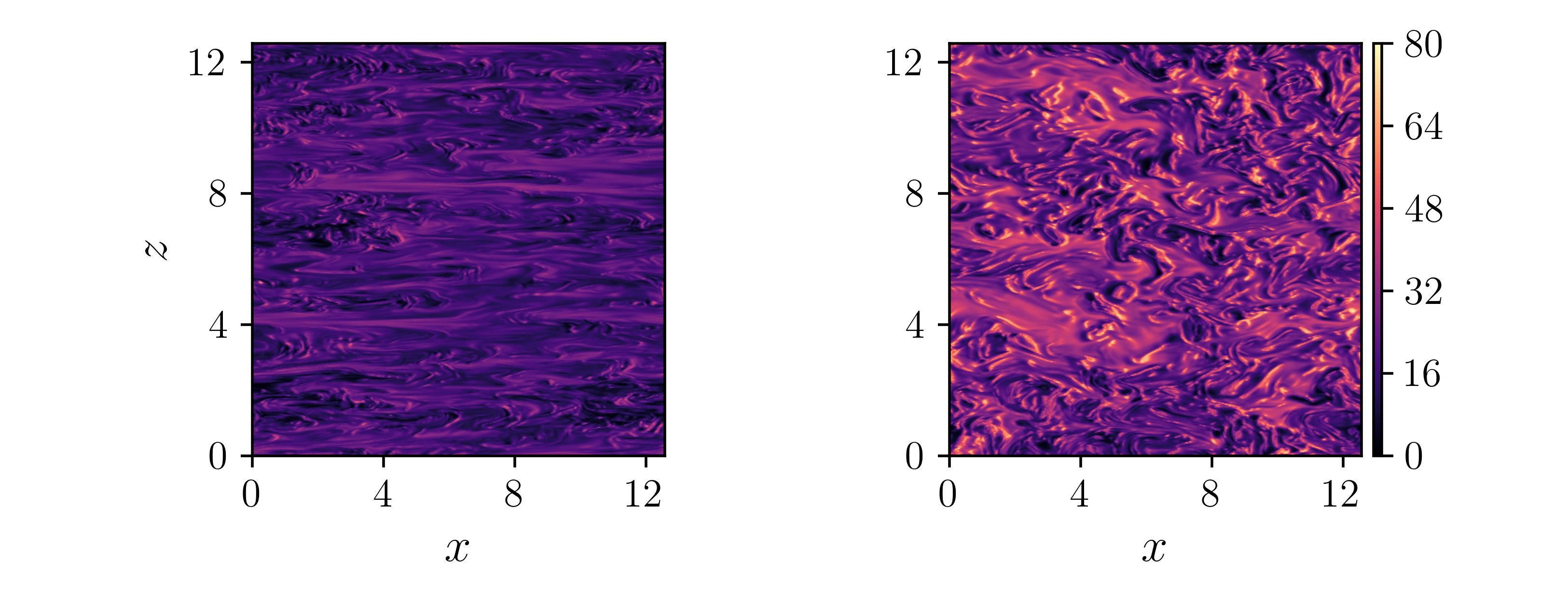}
\end{center}

\hspace{0.8in} (e) \hspace{1.6in} (f)

\begin{center}
\hspace{0.in}\includegraphics[draft=false,width=0.8\textwidth]{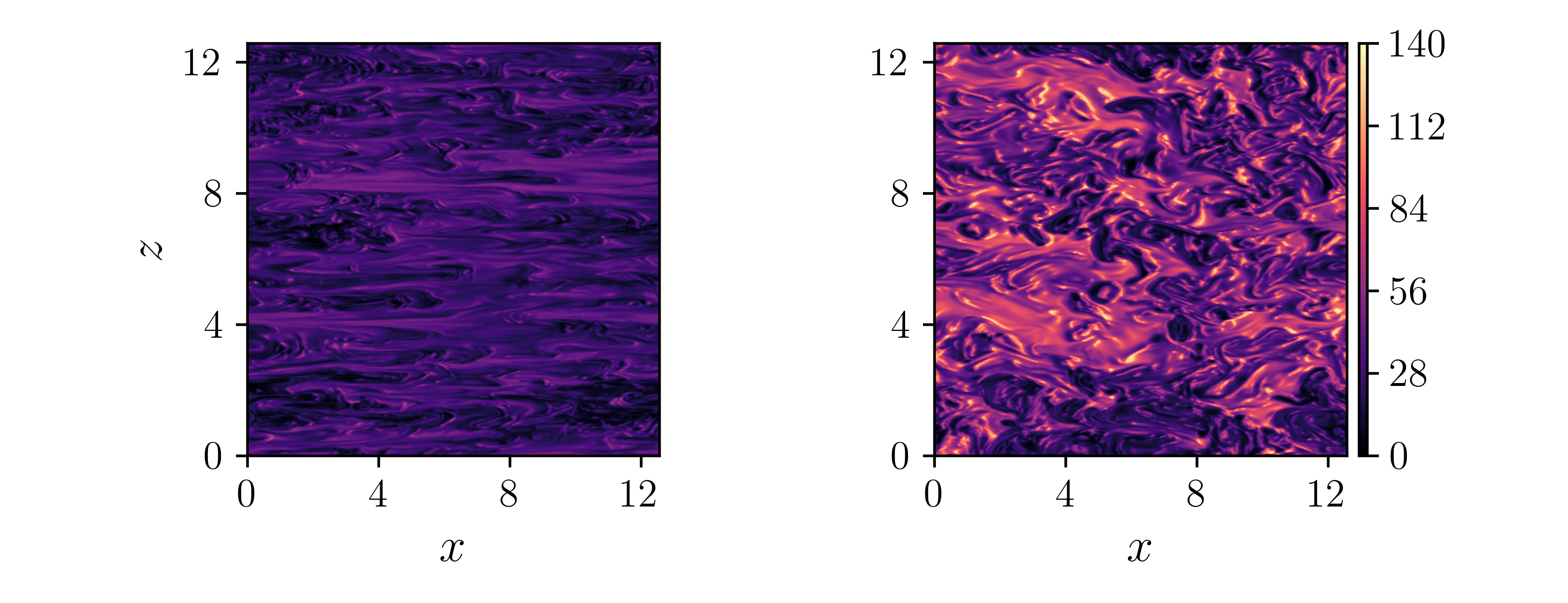}
\end{center}

\caption{Wall-parallel $(x,z)$ planes of isocontours of instantaneous (a)--(b): logarithmic volume ratio, $\zeta$, (c)--(d):  geodesic distance from the identity, $\kappa$, and (e)--(f): anisotropy index, $\xi$.
(a),(c), and (e): $y^+=15$. (b), (d), and (f): $y^+=180$ (centreline).}
\label{fig:metrics-snapshots}
\end{figure}

Figure \ref{fig:metrics-snapshots} shows isocontours of instantaneous values of $\zeta$, $\kappa$ and $\xi$ for two wall-parallel planes, $y^+ = 15$ and $y^+ = 180$. As a comparison, the fluctuating tensor $\mathsfi{C}_{xx}'$ obtained by the Reynolds decomposition (\ref{CReyDecomp}) and normalized by the local $\mathsfi{\overline{C}}_{xx}$ is shown in figure \ref{fig:Cxxp-snapshots}. 
We normalized  $\mathsfi{C}_{xx}'$ so that the fluctuations near the wall could be compared to those at the centreline, since $\mathsfi{\overline{C}}_{xx}$ differs by an order of magnitude between the two locations.
The isocontours of $\mathsfi{C}_{xx}'/\mathsfi{\overline{C}}_{xx}$, and in particular the negative values, are difficult to interpret since the correspondence to a physical deformation or a mathematical metric is unclear.   

Figures \ref{fig:metrics-snapshots}(a)--(b) show the logarithmic volume ratio, $\zeta$. This quantity is the logarithm of $\III{\mathsfbi{G}}$, which itself is qualitatively similar to $\I{\mathsfbi{G}}$ and hence we observe a strong visual resemblance between figures \ref{fig:rho-snapshots} and \ref{fig:metrics-snapshots}(a)--(b).  The colour scale in the former is logarithmic and is thus consistent with the linear scale in figure \ref{fig:metrics-snapshots}. Both figure \ref{fig:metrics-snapshots}(a) and (b) have predominantly negative values, indicating that the instantaneous volume is smaller than the volume of the mean conformation.  We also find regions of very high $\zeta$ adjacent to regions of very low values, especially at the centreline. This is partially a result of the lack of diffusion in the polymers since there is no direct mechanism for smoothing out shocks in the conformation tensor field. 

\begin{figure}
	
	\hspace{0.8in} (a) \hspace{1.6in} (b)

	\begin{center}
		\includegraphics[draft=false,width=0.8\textwidth]{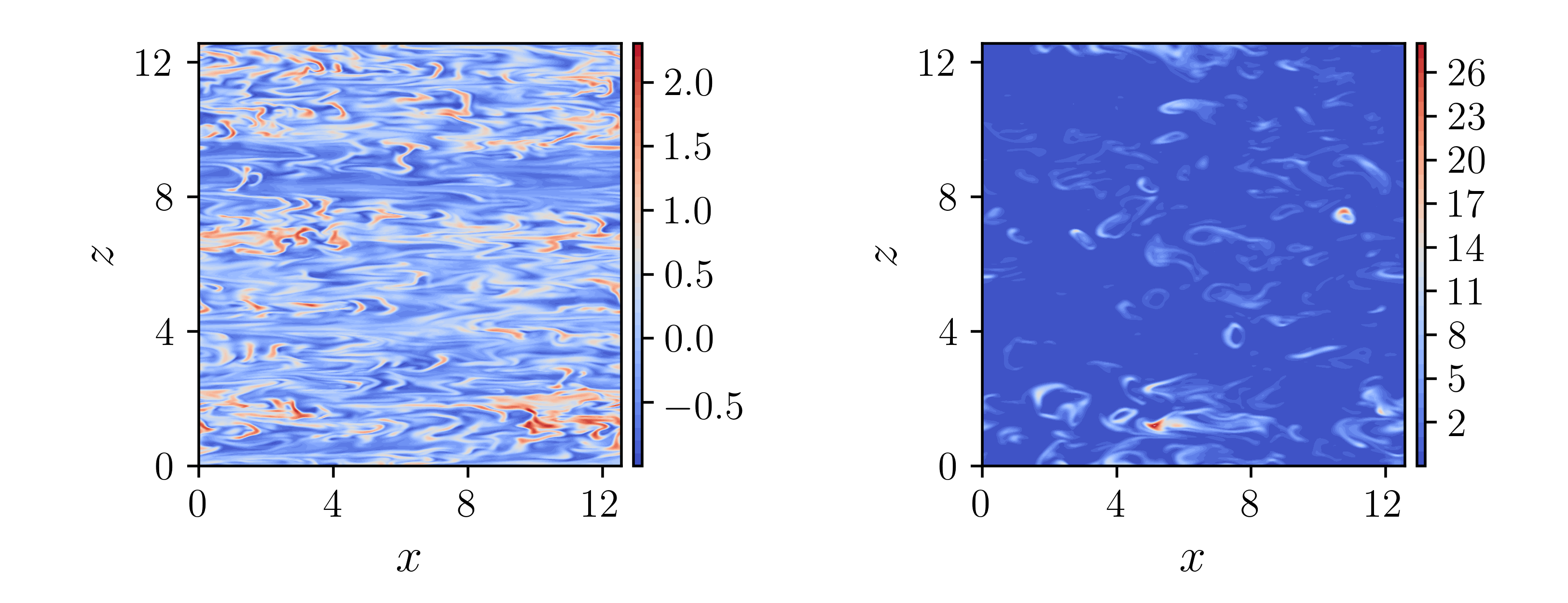}
	\end{center} 
	
	\caption{Wall-parallel $(x,z)$ planes of isocontours of instantaneous $\mathsfi{C}_{xx}'/\mathsfi{\overline{C}}_{xx}$ at (a) $y^+=15$ (b) $y^+=180$ (centreline), where $\mathsfi{\overline{C}}_{xx}(y^+=15) = 2.22 \times 10^3$ and $\mathsfi{\overline{C}}_{xx}(y^+=180) = 1.02 \times 10^2$. The limits of the  divergent colour map are set at the planar maxima and minima of $\mathsfi{C}_{xx}'/\mathsfi{\overline{C}}_{xx}$.}
	\label{fig:Cxxp-snapshots}
\end{figure}

Another important effect is that of memory: polymers are stretched near the wall where the shear is significant, and are then transported out to the centreline. If the the half-channel transit time, the ratio of the channel half-height to the RMS wall-normal velocity, is smaller than the polymer relaxation time, we expect to observe a footprint of the near-wall stretching all the way out to the centreline. In the present case, the relaxation time is two to three times the half-channel transit time. Since material points that are initially close are exponentially diverging in a turbulent flow \citep[see][for a study of Lagrangian stretching in Newtonian turbulent channel flow]{Johnson2017}, it is unsurprising to find adjacent regions of strongly and weakly stretched polymers. 

The reductionist explanation given above is useful for a basic understanding but is insufficient to account for other observed features of the flow. For example, the present considerations would suggest that the streamwise elongated shape of the isocontours of $\zeta$ near the wall would lead to a similar shape at the centreline. However, this is manifestly not the case. Instead, the $\zeta$ field appears to generate, on the whole, highly curved isocontours at the centre of the channel.

The measure $\zeta$ does not distinguish between volume-preserving deformations. For example, $\zeta$ does not distinguish between $\mathsfbi{C}$ and $(\det\,\mathsfbi{C})\mathsfbi{A}$ for \emph{any} $\mathsfbi{A}$ with determinant $= 1$. In particular, $\zeta = 0$, does not imply  $\mathsfbi{C} = \mathsfbi{\overline{C}}$. In order to identify regions where $\mathsfbi{C} = \mathsfbi{\overline{C}}$ is true, and quantify the deviation when it is not, we use the squared distance away from the origin ($\mathsfbi{I}$) along the manifold, $\kappa$. Figure \ref{fig:metrics-snapshots}(c)--(d) shows isocontours of instantaneous $\kappa$. Most of the conformation tensor field is significantly far away, in the sense of distance along the manifold, from $\mathsfbi{\overline{C}}$. However, regions where $\mathsfbi{\overline{C}}$ is a good representation of $\mathsfbi{C}$ are interspersed between regions where $\kappa$ is large. This behaviour is true both in the near-wall region as well as the channel centre but more so in the latter. In contrast, it is well-known that kinetic energy fluctuations are weakest at the centreline in a Newtonian channel flow. A different behaviour for the polymers is unsurprising since, due to the strong memory effect, $\mathsfbi{C}$ at each point is strongly dependent on the Lagrangian path that is obtained by a pull-back of the particular Eulerian point of interest.

Figures \ref{fig:metrics-snapshots}(e)--(f) show isocontours of instantaneous $\xi$, the anisotropy index. This index shows how close the shape of instantaneous conformation tensor is to the shape of the mean conformation tensor, irrespective of volumetric changes. The visual resemblance of $\kappa$ and $\xi$ suggests that deformations to the mean conformation are largely anisotropic, or in other words, lead to shape change.

\begin{figure} 

\centering
\includegraphics[draft=false,width=0.45\textwidth]{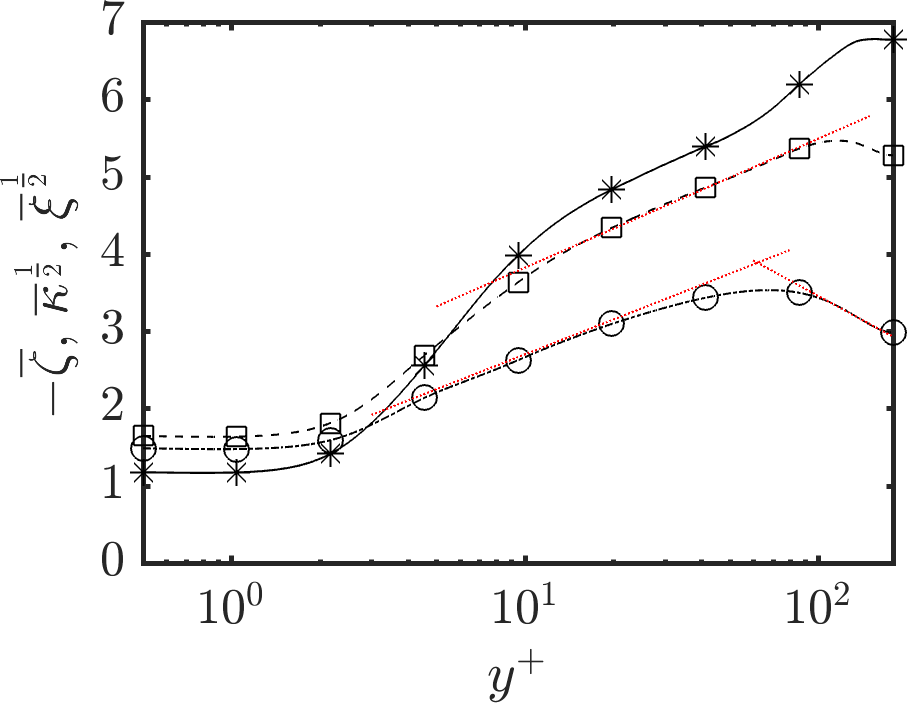}
\caption{
Mean scalar measures, plotted in equivalent dimensions, as functions of $y^+$. The solid line with star symbols (\linesolidsSym)  is minus the mean volume ratio: $-\overline{\zeta} = -\langle \zeta \rangle$, the dashed line  with square symbols (\linedashedSym) is the square-root of the mean geodesic distance from the identity: $\overline{\kappa}^{\frac{1}{2}} = \sqrt{\langle \kappa \rangle}$, the dashed-dot line with circle symbols (\linedshdotSym) is the square-root of the mean anisotropy index: $\overline{\xi}^{\frac{1}{2}} = \sqrt{\langle \xi \rangle}$. 
Red dotted lines ({\color{red}\linedotted}) are logarithmic fits to the profiles: the fit to the $\overline{\kappa}^{\frac{1}{2}}$ profile (\linedashedSym)  is given by $0.725\log\,y^+ +2.15$, the fits to the $\overline{\xi}^{\frac{1}{2}}$ profile (\linedshdotSym) are given
by $0.65\log\,y^+ + 1.20$ and $-0.9\log\,y^+ + 7.6$. Note that the symbols are identifiers and are thus only a small subset of all the data points used in the line plots.}
\label{fig:metrics-profiles}
\end{figure}

Figure \ref{fig:metrics-profiles} shows the mean  values of $\zeta$ , $\kappa$ and $\xi$ in dimensions of distance along the manifold. These statistics were generated using 225 convective time units. We checked for convergence by halving the number of samples. This process led to only minor deviations in the results, with no material significance to the discussion that follows. As was inferred earlier, the average logarithmic volume ratio is negative throughout the channel and is monotonically decreasing towards the centreline where it becomes roughly constant, similar to the behaviour very near the wall  $y^+ \lesssim 2$. This behaviour of the mean logarithmic volume being smaller than the
volume of the mean is consistent with the `swelling' problem
associated with the arithmetic mean of positive-definite tensors that
has been previously reported in the literature \citep{Arsigny2007}. It also suggests that, although an arithmetic mean of the conformation tensor may be unavoidable for modelling in the averaged equations, it may not be the most representative conformation tensor for deducing the most likely physical deformation of the polymers.
A more extensive study, beyond the scope of the present work, is required to determine better alternatives to the arithmetic mean.

The square-root mean squared distance from the origin along the manifold, $\overline{\kappa}^{\frac{1}{2}}$, is logarithmically increasing up to close to the centreline but peaks at $y^+ \approx 100$. The anisotropy, $\overline{\xi}^{\frac{1}{2}}$, shows logarithmic increase over a small range $3 \lesssim y^+ \lesssim 20$, peaking at $y^+ \approx 60$, but then shows a logarithmic decrease towards the centreline. The logarithmic behaviour in these quantities, especially in $\overline{\kappa}^{\frac{1}{2}}$ where the behaviour extends over a significant range, resembles the behaviour that appears in the mean velocity as well as in the statistical moments and two-point correlations of the velocity fluctuations in wall-bounded shear flows \citep{Meneveau2013,Yang2016}.

If we momentarily accept the simplified picture of polymers being deformed closer to the wall in an `active' region and then passively transported out to the centreline of the channel, then these results indicate that the active region of the channel exists all the way up to $y^+\approx 100$. Beyond this region, the stretching of polymers weakens and thus $\overline{\kappa}^{\frac{1}{2}}$ decreases monotonically. The active region involves a region of logarithmic increase, and so we can write $\overline{\kappa}^{\frac{1}{2}} = a_1 \int y^{-1}\,\text{d}y + a_0$ for constants $a_1$ and $a_0$. If the stretching at each wall-normal station is actually additive, the polymers are undergo deformation that on average leads to a diminishing increment in the deformation. This behaviour is consistent with the velocity gradients weakening with wall-normal distance. At $y^+\approx 100$, the velocity gradients then either weaken or act on such long time scales that polymers relax quickly enough not to retain any additional deformation.  

\begin{figure}
\centering
\includegraphics[draft=false,width=0.8\textwidth]{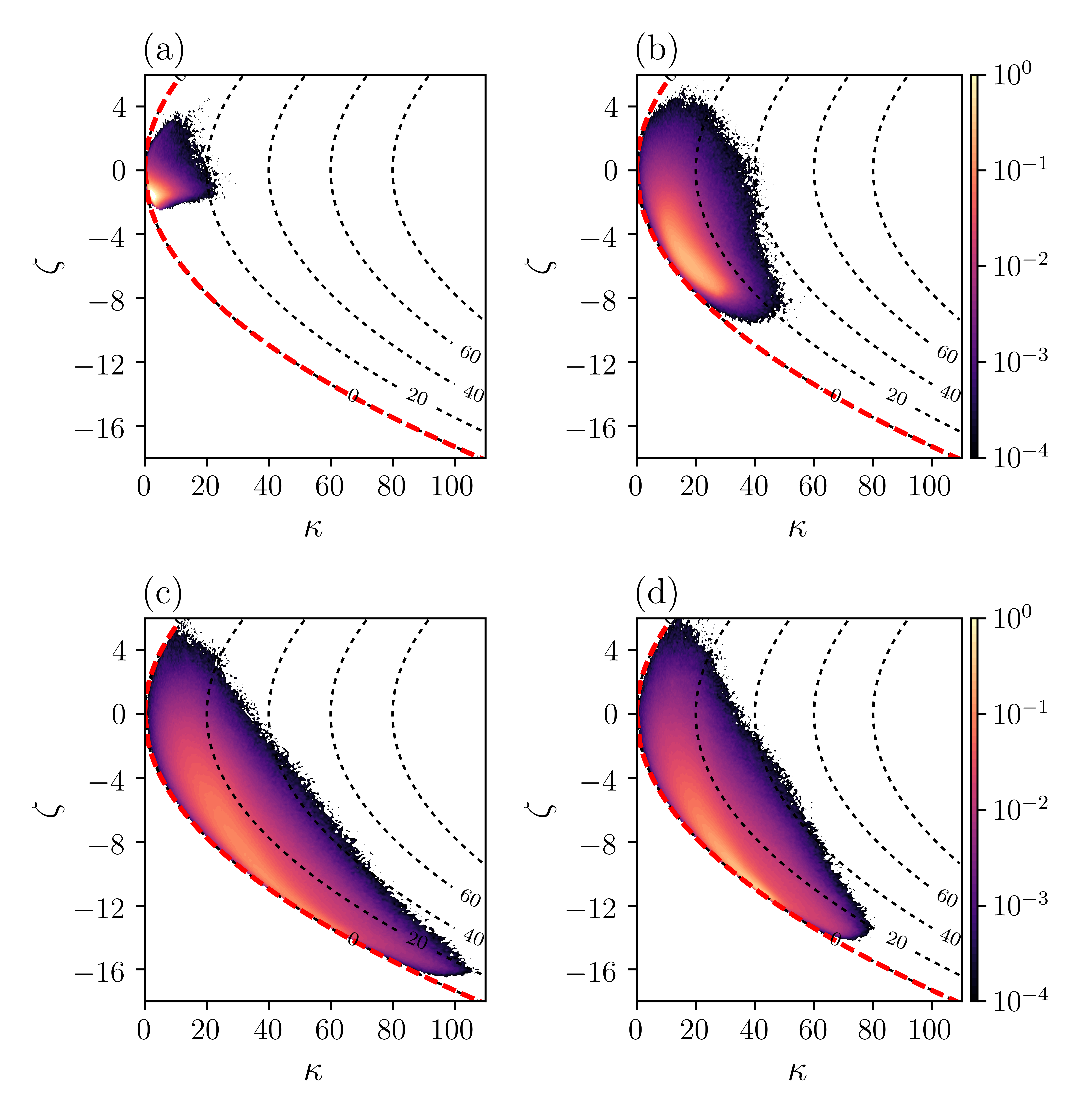}
\caption{Joint probability density functions (JPDF) of the logarithmic volume ratio, $\zeta$, and the geodesic distance from the identity, $\kappa$, at four different wall-normal locations: (a) $y^+=2$,
(b) $y^+=15$, (c) $y^+=100$ and (d) $y^+=180$ (centreline). Dotted lines (\linedotted) are isocontours of the anisotropy index, $\xi$. The thick red dashed line ({\color{red}\linedashedthick}) denotes the realizability bound, $\kappa = \frac{1}{3}\zeta^2$, derived in (\ref{admissible_zeta12_bounds}), which coincides with the zero anisotropy index isocontour ($\xi=0$).}
\label{fig:metrics-jpdf}
\end{figure}

In order to quantify the fluctuations observed in figure \ref{fig:metrics-snapshots} in more detail, we calculated the joint probability density function (JPDF) for $\zeta$ and $\kappa$ using 12 snapshots evenly spaced over 120 convective time units. The JPDF, at four different wall-normal locations, are shown in figure \ref{fig:metrics-jpdf} along with isocontours of $\xi$, which is purely a function of $\zeta$ and $\kappa$. 

The JPDF are non-zero primarily on the lower half of the realizability region, which is consistent with the isocontours in figure \ref{fig:metrics-snapshots}(a)--(b) and $\overline{\zeta} < 0$ throughout the channel in figure \ref{fig:metrics-profiles}. In addition, the isocontours tend to concentrate along the isotropy line (thick red dashed line) that was derived in (\ref{admissible_zeta12_bounds}) as a realizability bound. However, with the exception of the centreline, the most probable $(\zeta,\kappa)$ are located away from the isotropy line. This implies that the most likely conformation tensor away from the centreline does not have the same shape as the local mean conformation tensor. Although the most likely conformation tensor at the centreline assumes the shape of the mean conformation tensor, the JPDF at the centreline occupies a greater area in the $(\zeta,\kappa)$ plane than at $y^+=2$ or $y^+=15$ which implies a greater degree of uncertainty. In addition, the most likely conformation tensor, as determined by the peak of the JPDF, is further away from the mean than at any other wall-normal location.

The JPDF indicate that the most intermittent region of the flow, determined by the most extreme excursions away from the identity on $\Pos_3$, do not occur near the wall or at the centreline. This can be seen in figure \ref{fig:metrics-jpdf}, where the JPDF at $y^+=100$ shows events with up to $\kappa = 100$. This behaviour is consistent with the peak $\overline{\kappa}$ occurring away from the centreline  in figure \ref{fig:metrics-profiles}.

\begin{figure}
\hspace{0.6in} (a) \hspace{1.8in} (b)

\centering
\includegraphics[draft=false,width=0.75\textwidth]{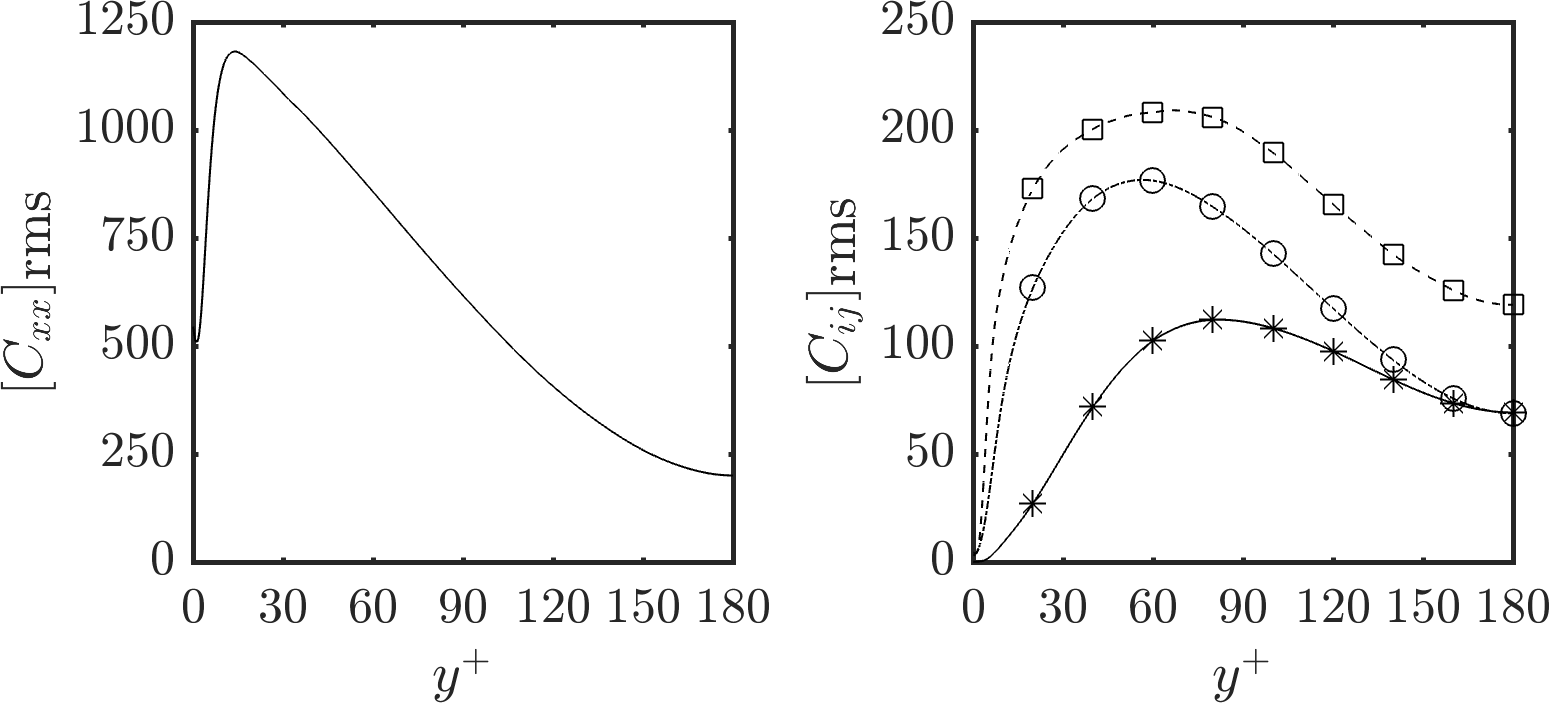}
\caption{Root-mean-square profiles of $\mathsfbi{C}'$, $[\mathsfi{C}_{ij}]_{\text{rms}} = \sqrt{\langle \mathsfi{C}_{ij}^2 \rangle - \langle \mathsfi{C}_{ij} \rangle^2}$ based on the Reynolds decomposition (\ref{CReyDecomp}). (a) $[\mathsfi{C}_{xx}]_{\text{rms}}$ (b) the solid line with star symbols (\linesolidsSym) is $[\mathsfi{C}_{yy}]_{\text{rms}}$, the dashed line with square symbols  (\linedashedSym) is $[\mathsfi{C}_{zz}]_{\text{rms}}$, and the dashed-dot line with circle symbols
		(\linedshdotSym) is $[\mathsfi{C}_{xy}]_{\text{rms}}$. Note that the symbols in (b) are identifiers  and are thus only a small subset of all the data points used in the line plots. }
\label{fig:C-rms-profiles}
\end{figure}

Finally the RMS of $\mathsfbi{C}'$, defined according to the Reynolds decomposition (\ref{CReyDecomp}), is shown in figure \ref{fig:C-rms-profiles} for comparison to the present approach. The protocol used to obtain these quantities was the same as that used to obtain the mean conformation tensor in figure \ref{fig:C-mean-profiles}.   The RMS of $\mathsfbi{C}'$ are of similar magnitude to the components of $\langle \mathsfbi{C} \rangle$. In fact, the peak RMS of $\mathsfi{C}_{yy}$, $\mathsfi{C}_{zz}$ and $\mathsfi{C}_{xy}$ are larger than their respective mean values. Interestingly, the peak fluctuating deformation found at $y^+ \approx 100$ using our present framework is not discernible from the RMS fluctuations. Different components of the RMS tensor peak at different locations in the channel with $[\mathsfi{C}_{yy}]_{\text{rms}}$ showing a peak that is closest to $y^+ = 100$. The RMS quantities only show the component-wise behaviour of the conformation tensor, and hence are not indicative of the total polymer deformation. A more appropriate quantity to evaluate in the context of $\mathsfbi{C}'$ would then be the JPDF of all six independent components of $\mathsfbi{C}'$. Owing to high-dimensionality, this characterization is more difficult to both calculate and analyse. The scalar measures suggested in the present work provide a good alternative to such a characterization.

\section{Conclusion}
\label{sec:conclusion} 
We have developed a geometric decomposition, given in (\ref{NecessaryDecomposition}), that overcomes the difficulties associated with the traditional Reynolds decomposition of $\mathsfbi{C}$.
The geometric decomposition yields a conformation tensor, $\mathsfbi{G}$, that describes the deformation of the polymer with respect to the mean deformation.
We characterized the fluctuations in $\mathsfbi{G}$ by using a geometry specifically constructed for $\Pos_3$ and obtained three scalar measures: the logarithmic volume ratio, $\zeta$, given in (\ref{zeta1defn}), the squared geodesic distance of the perturbation conformation tensor away from the origin, $\kappa$, given in (\ref{zeta2defn}), and the anisotropy index, $\xi$ given in (\ref{chiDefn1}), defined as the squared geodesic distance to the closest isotropic tensor.
The average values and JPDF of these scalar measures provided interesting insights about the fluctuating polymer deformation that are not readily available from a Reynolds decomposition of $\mathsfbi{C}$. These insights include the following:
\begin{enumerate}
\item The anisotropy in $\mathsfbi{\overline{C}}$, measured as geodesic distance away from $\mathsfbi{I}$ on $\Pos_3$, decreases logarithmically from $y^+ = 5$ to close to the centreline.
\item The mean conformation tensor tends to be significantly different than the most likely conformation tensor observed in the flow.
\item The mean polymer deformation, measured in terms of $\overline{\kappa}$, increases logarithmically from $y^+=10$ and peaks at $y^+\approx 100$. 
\item As evidenced by the JPDF of $\kappa$, the peak turbulence intensity in the polymers occurs in between the wall and centreline, at approximately $y^+\approx 100$.
\end{enumerate}
The universality of the trends mentioned above, and others documented in the present work, and their connection to  larger issues in viscoelastic turbulence are open questions. The framework we have developed can be used
to probe the dynamics in viscoelastic turbulence beyond channel flow and can also be exploited for
developing and benchmarking reduced-order models for viscoelastic turbulence.  
The approach can also be adapted to other similar problems, for example in the analysis of deforming droplets in turbulence using a model based on the droplet conformation tensor \citep{Maffettone1998,Biferale2014}.

An important, open question that needs to be resolved in future work is the relationship between the fluctuating conformation tensor $\mathsfbi{G}$ and elastic energy of the polymers. In contrast to the clear meaning of the kinetic energy associated with fluctuating velocity field, a deeper understanding of the elastic energy and its relation to $\mathsfbi{G}$ and the scalar measures introduced in the present work is unavailable. The attainment of such an understanding is partially hindered by the myriad of constitutive models prevalent in the literature \citep{Beris1994}. Instead of using the details of a particular constitutive model, the aim of the present work was to maintain as much generality as possible by exploiting the mathematical structure of $\mathsfbi{G}$ to characterize the fluctuating polymer deformation.

 \appendix

\section{Riemannian structure of the set of positive-definite matrices}

The theoretical results presented in this section on the geometric structure of $\Pos_n$ are standard with
detailed accounts available in pp. 322--339 of \citet{Lang2001} and also pp. 201--235 of \citet{Bhatia2015}.

\subsection{Riemannian metric}
We can define an inner product $(\cdot,\cdot):\mathbb{R}^{n\times n}\times\mathbb{R}^{n\times n} \rightarrow \mathbb{R}$
\begin{align}
(\mathsfbi{A},\mathsfbi{B})_{\mathsfbi{X}} = \tr\,
\left( \mathsfbi{X}^{-1}\bcdot \mathsfbi{A}^{\mathsf{T}} \bcdot \mathsfbi{X}^{-1} \bcdot \mathsfbi{B} \right)
\label{PosInnerProd}
\end{align} 
for $\mathsfbi{A},\mathsfbi{B} \in \mathbb{R}^{n \times n}$ where $\mathsfbi{X} \in \Pos_n$ is fixed. When $\mathsfbi{X}=\mathsfbi{I}$, (\ref{PosInnerProd}) reduces to the definition of the standard Frobenius inner product.
The guaranteed factorization $\mathsfbi{X}=\mathsfbi{X}^{\frac{1}{2}} \bcdot \mathsfbi{X}^{\frac{1}{2}}$ and the cyclical
property of the trace ensures the positivity of $\|\mathsfbi{A}\|_{\mathsfbi{X}}$, while the remaining properties 
of the inner product and the norm follow-on from the standard Frobenius theory. The space 
$\mathbb{R}^{n\times n}$ is a Hilbert space when equipped with such an inner product and the norm induced
by it: $\|\mathsfbi{A}\|_{\mathsfbi{X}} = \sqrt{(\mathsfbi{A},\mathsfbi{A})_{\mathsfbi{X}}}$. The subset of 
$\mathbb{R}^{n\times n}$ consisting of symmetric matrices forms a vector space, $\Sym_n$, and can
also be Hilbertized under the inner product (\ref{PosInnerProd}). $\Pos_n$ is an open subset of 
$\mathbb{R}^{n\times n}$ (and also $\Sym_n$) in the $\|\cdot\|_{\mathsfbi{X}}$ metric and is thus a 
(smooth) manifold.

The tangent bundle of $\Pos_n$, which consists of the manifold $\Pos_n$ equipped with a tangent space 
$T_{\mathsfbi{X}} \Pos_n$ at each point $\mathsfbi{X}$ of $\Pos_n$, provides a natural projection that can be used
to study the geometry of $\Pos_n$. A simple argument shows that the tangent space at each point of $\Pos_n$ coincides with $\Sym_n$ (Let $\mathsfbi{Y}$ be defined by $\mathsfbi{Y}=\mathsfbi{X} + \varepsilon \mathsfbi{S}$, for some $\mathsfbi{X}\in\Pos_n$,
$\varepsilon \in \mathbb{R}$, $\mathsfbi{S} \in \Sym_n$. By Weyl's inequality, there exists some 
$\varepsilon > 0$ sufficiently small such that $\mathsfbi{Y} \in \Pos_n$. This implies that 
$\Sym_n \subseteq T_{\mathsfbi{X}} \Pos_n$. Since $\mathsfbi{Y}\notin\Pos_n$ for any $\varepsilon \neq 0$ and
$\mathsfbi{S} \notin \Sym_n$, we have $\Sym_n = T_{\mathsfbi{X}} \Pos_n$.). The latter result is the geometric underpinning of numerical algorithms
that time march the conformation tensor by translations of $\mathsfbi{C}$ by symmetric matrices (the right-hand
side of the evolution equation for $\mathsfbi{C}$). 

A manifold $\textbf{M}$ equipped with a scalar product over $T_{\mathsfbi{X}} \textbf{M}$ for each $\mathsfbi{X} \in \textbf{M}$ is a Riemannian manifold.
The set of such scalar products is called the Riemannian metric of the manifold. $\Pos_n$ is a Riemannian manifold with Riemannian metric given by \citep{Lang2001,Bhatia2015}
\begin{align}
g=\left\{ (\cdot,\cdot)_{\mathsfbi{X}} | \mathsfbi{X} \in \Pos_n  \right\}
\label{RiemmanianMetric}
\end{align}
with the understanding that the scalar product on $T_{\mathsfbi{X}} \Pos_n$ is $(\cdot,\cdot)_{\mathsfbi{X}} \in g$
and the domain of $(\cdot,\cdot)_{\mathsfbi{X}}$ is restricted to $\Sym_n = T_{\mathsfbi{X}} \Pos_n$. An 
infinitesimal distance around the point $\mathsfbi{X}$ on the manifold is given by 
\begin{align}
ds^2 = \| \text{d} \mathsfbi{X} \|_{\mathsfbi{X}}^2 = \tr\,\left[(\mathsfbi{X}^{-1}\bcdot \text{d} \mathsfbi{X})^2\right]
\label{ds2Manifold}
\end{align}
The metric given by (\ref{RiemmanianMetric}) ensures that distances between points 
$\mathsfbi{X}, \mathsfbi{Y} \in \Pos_n$ along the manifold calculated using (\ref{ds2Manifold}) are invariant under
the action $\glaction{\mathsfbi{A}}{\cdot}$ of any $\mathsfbi{A}\in\GL_n$, i.e. invariant to transformations such as (\ref{NecessaryDecomposition}).

\subsection{Geodesic curves and distances}

Consider a parameterized curve on $\Pos_n$ connecting points $\mathsfbi{X},\mathsfbi{Y} \in \Pos_n$, i.e.
$\mathsfbi{P}: [0,1] \rightarrow \Pos_n$ with $\mathsfbi{P}(0) = \mathsfbi{X}$ and $\mathsfbi{P}(1) = \mathsfbi{Y}$. The distance,
in the sense of the metric $g$, traversed on the manifold along the curve $\mathsfbi{P}=\mathsfbi{P}(r)$ is given by
\begin{align}
\ell_{\mathsfbi{P}}(r) \equiv \int_0^r \left\|  \frac{\text{d} \mathsfbi{P}(r')}{\text{d} t}  \right\|_{\mathsfbi{P}(r')} \, \text{d}r'.
\end{align}
$\ell_{\mathsfbi{P}}$ has an attractive property in that it is invariant under affine transformations.

\begin{lem}[Affine invariance]
For each positive-definite $\mathsfbi{A}$ and differentiable path $\mathsfbi{P}$ on the Riemannian manifold of
positive-definite matrices, we have
\begin{align}
\ell_{\mathsfbi{P}}=\ell_{\glaction{\mathsfbi{A}}{\mathsfbi{P}}}.
\end{align}
\end{lem}
\begin{proof}
 \S 6.1.1 in \citet{Bhatia2015}. 
\end{proof}

We call a curve $\mathsfbi{P}(r)$ on $\Pos_n$ that minimizes $\ell_{\mathsfbi{P}}(1)$ a \emph{geodesic curve} 
connecting $\mathsfbi{X}$ and $\mathsfbi{Y}$. In general, the existence and/or uniqueness of a geodesic curve is not
guaranteed. We also define $d(\mathsfbi{X},\mathsfbi{Y})$, the \emph{geodesic distance} between $\mathsfbi{X}$ and 
$\mathsfbi{Y}$ as the infimum of $\ell_{\mathsfbi{P}}(1)$ over all possible curves $\mathsfbi{P}$ connecting $\mathsfbi{X}$
and $\mathsfbi{Y}$
\begin{align}
d(\mathsfbi{X},\mathsfbi{Y}) \equiv
\inf_{\mathsfbi{P}}\left\{ \ell_{\mathsfbi{P}}(1) | \mathsfbi{P}(r) \in \Pos_n, \mathsfbi{P}(0) = \mathsfbi{X}, \mathsfbi{P}(1)=\mathsfbi{Y}  \right\}.
\end{align}
A corollary of affine invariance is that 
$d(\mathsfbi{X},\mathsfbi{Y}) = d(\glaction{\mathsfbi{A}}{\mathsfbi{X}},\glaction{\mathsfbi{A}}{\mathsfbi{Y}})$.

It turns out that the existence and uniqueness of geodesics is guaranteed on $\Pos_n$.
Furthermore, we can obtain analytical expressions for these geodesics. Following \citet{Bhatia2015}, we present three key
theorems that allow this construction.

\begin{thm}[Exponential metric increasing property]
For any two real symmetric $\bld{\mathcal{X}}$ and $\bld{\mathcal{Y}}$
\begin{align}
d(\e{\bld{\mathcal{X}}},\e{\bld{\mathcal{Y}}}) \leq \| \bld{\mathcal{X}} - \bld{\mathcal{Y}} \|_{\mathsfbi{I}}
\label{EMI}
\end{align}
where we note that $\e{\bld{\mathcal{X}}}$, $\e{\bld{\mathcal{Y}}}$ are positive-definite matrices.
\end{thm}
\begin{proof}
 \S  XII.2 in \citet{Lang2001} and {\S}6.1.4 in \citet{Bhatia2015}. 
\end{proof}

Equality is achieved in (\ref{EMI}) when $\mathsfbi{X}$ and $\mathsfbi{Y}$ commute and we can also parameterize the
geodesic in this case, as expressed in the proposition below.

\begin{prop}
Let $\mathsfbi{X} = \e{\bld{\mathcal{X}}}$ and $\mathsfbi{Y} = \e{\bld{\mathcal{Y}}}$ be positive-definite matrices
such that $\mathsfbi{X}\bcdot\mathsfbi{Y}=\mathsfbi{Y}\bcdot \mathsfbi{X}$. Then, the exponential function maps the line segment
\begin{align}
(1-r)\bld{\mathcal{X}} + r \bld{\mathcal{Y}}, \qquad 0 \leq r \leq 1
\end{align}
in the Euclidean space of symmetric matrices to the geodesic between $\mathsfbi{X}$ and $\mathsfbi{Y}$ on the
Riemannian manifold of positive-definite matrices and
\begin{align}
d(\mathsfbi{X},\mathsfbi{Y}) = \| \bld{\mathcal{X}} - \bld{\mathcal{Y}} \|_{\mathsfbi{I}}
\end{align} 
\end{prop}
\begin{proof}
Chapter 6 in \citet{Bhatia2015}.
\end{proof}

Finally, using the affine invariance property of the Riemannian metric and noting that $\mathsfbi{I}$ commutes
with every element of $\Pos_n$, one can prove the following theorem.

\begin{thm}
Let $\mathsfbi{X}$ and $\mathsfbi{Y}$ be positive-definite matrices. There exists a unique geodesic 
$\mathsfbi{X}\#_r \mathsfbi{Y}$ on the Riemannian manifold of positive-definite matrices that joins $\mathsfbi{X}$ and 
$\mathsfbi{Y}$ with the following parametrisation
\begin{align}
\mathsfbi{X}\#_r \mathsfbi{Y} = 
\mathsfbi{X}^{\frac{1}{2}} \bcdot 
\left(\mathsfbi{X}^{-\frac{1}{2}}\bcdot \mathsfbi{Y}\bcdot \mathsfbi{X}^{-\frac{1}{2}} \right)^r \bcdot
\mathsfbi{X}^{\frac{1}{2}}
\end{align} 
which is natural in the sense that
\begin{align}
d(\mathsfbi{X},\mathsfbi{X}\#_r \mathsfbi{Y}) = r d(\mathsfbi{X},\mathsfbi{Y})
\end{align} 
for each $r \in \mathbb{R}$. Furthermore, we have
\begin{align}
d(\mathsfbi{X},\mathsfbi{Y}) = 
\left\| \log\, \left( \mathsfbi{X}^{-\frac{1}{2}}\bcdot\mathsfbi{Y}\bcdot \mathsfbi{X}^{-\frac{1}{2}}\right)
\right\|_{\mathsfbi{I}} 
= \left[\sum_{i=1}^3 \left(\log\, \sigma_i\left( \mathsfbi{X}^{-1}\bcdot \mathsfbi{Y} \right)\right)^2\right]^{\frac{1}{2}}
\end{align} 
\end{thm}
\begin{proof}
Chapter 6 in \citet{Bhatia2015}.
\end{proof} 

 \section{Details of the numerical evolution of the conformation tensor equations}
  
The conformation tensor field is solved on a temporal grid that is staggered by half a time step with respect to that of the velocity field. We use linear interpolation to transfer between the two temporal grids. The conformation tensor is time-marched using an equal sub-step second-order Runge--Kutta scheme. The advection and  stretching terms  employ an Adams-Bashforth discretization at each sub-step.  Following \citet{Dubief2005}, we ensure that the polymer does not exceed its maximum extensibility by  using a semi-implicit approach for the relaxation term in the conformation tensor.

The advection term in the conformation tensor is second-order accurate in space with the exception of a small number points where it is only first-order accurate: $\sim 0.1-0.3$\% in the $x$ and $y$ directions   and $\sim 3-5$\% in the $z$-direction. This change in order is due to the special treatment of the advection term in the   conformation tensor that is needed to avoid numerical issues which lead to the conformation tensor losing positive-definiteness. The special treatment is an adaptation of the slope-limiting  approach of \citet{Vaithianathan2006} (see also  \citet{Dallas2010} for an implementation). The approach of \citet{Vaithianathan2006} requires the evaluation of three schemes, utilizing forward, backward and centred stencils, to approximate the flux at the boundaries of each computational cell, at each time step. The scheme that maximizes the eigenvalues of the conformation tensor at the boundaries is then chosen. We modify this approach by evaluating the following schemes in order, and choosing the first one that yields a positive-definite tensor at the boundaries: (a) centred stencil (b) upwind biased stencil (c) downwind biased stencil and (d) first-order approximation that equates the conformation tensor at the cell-centre and the boundary. Our approach is mathematically consistent with \citet{Vaithianathan2006} but is computationally more efficient and defers to the unbiased approximation when possible. In practice, we find that case (a) is sufficient for the vast majority of the points in  the domain ($> 90$\%).  
  
The evolution equation for the conformation tensor has no associated boundary conditions since it is hyperbolic. The conformation tensor at the walls, where no derivatives of the conformation tensor are needed, is then explicitly marched in time.

 \section*{Acknowledgements}

The authors would like to acknowledge the Maryland Advanced Research Computing Centre (MARCC) for computing
time, S. J. Lee
and J. Lee for help in the validation of the numerical code, and the referees for their valuable comments.
This work is sponsored in part by the National Science Foundation under grants CBET--1652244, CBET--1511937 and PHY--1125915.
\bibliographystyle{jfm}
\bibliography{manuscript}

\end{document}